\documentclass[a4paper,twocolumn,11pt,unpublished]{quantumarticle}
\pdfoutput=1
\usepackage[utf8]{inputenc}
\usepackage[english]{babel}
\usepackage[T1]{fontenc}
\usepackage{amsmath}
\usepackage{hyperref}

\usepackage{tikz}
\usepackage{booktabs}
\usepackage{amsthm}
\usepackage{amssymb}
\usepackage{mathtools}

\def\ket#1{{\lvert}#1\rangle}
\def\bra#1{{\langle}#1\rvert}

\def\proj#1{\ket{#1}\bra{#1}}
\def\norm#1{\left\| #1 \right\|}
\def\abs#1{\left| #1 \right|}

\newcommand{\N}{\Gamma}
\newcommand{\E}{{\overrightarrow{E}}}
\newcommand{\R}{\overrightarrow{{\cal R}}}
\newcommand{\w}{{\sf w}}
\newcommand{\p}{{\sf p}}
\newcommand{\f}{\theta}
\newcommand{\st}{$s$-$t$ }
\newcommand{\sm}{$s$-$M$ }
\newcommand{\ssm}{$\sigma$-$M$ }

\newcommand{\B}{{\cal B}}
\newcommand{\palt}{\p^{\sf alt} }
\newcommand{\falt}{\f^{\sf alt} }
\newcommand{\Aalt}{{\cal A}^{\sf alt}}

\newtheorem{theorem}{Theorem}[section]
\newtheorem{lemma}[theorem]{Lemma}
\newtheorem{corollary}[theorem]{Corollary}
\newtheorem{definition}[theorem]{Definition}

\newcommand{\eq}[1]{\hyperref[eq:#1]{(\ref*{eq:#1})}}
\newcommand{\thm}[1]{\hyperref[thm:#1]{Theorem~\ref*{thm:#1}}}
\newcommand{\lem}[1]{\hyperref[lem:#1]{Lemma~\ref*{lem:#1}}}
\newcommand{\cor}[1]{\hyperref[cor:#1]{Corollary~\ref*{cor:#1}}}
\newcommand{\defin}[1]{\hyperref[def:#1]{Definition~\ref*{def:#1}}}
\newcommand{\fig}[1]{\hyperref[fig:#1]{Figure~\ref*{fig:#1}}}
\newcommand{\sect}[1]{\hyperref[sec:#1]{Section~\ref*{sec:#1}}}
\newcommand{\app}[1]{\hyperref[app:#1]{Appendix~\ref*{app:#1}}}
\newcommand{\tabl}[1]{\hyperref[tab:#1]{Table~\ref*{tab:#1}}}

\begin{document}

\title{Quantum Walks for Chemical Reaction Networks}

\author{Seenivasan Hariharan}
\affiliation{Institute of Theoretical Physics, University of Amsterdam, Amsterdam, The Netherlands}
\affiliation{QuSoft, CWI, Amsterdam, The Netherlands}
\orcid{0000-0003-4509-8454}
\email{hseeni@gmail.com}

\author{Sebastian Zur}
\affiliation{IRIF \& CNRS, Paris, France}
\thanks{The majority of this work was conducted while SZ was affiliated with CWI \& QuSoft, the Netherlands.}

\author{Sachin Kinge}
\affiliation{Toyota Motor Europe, Materials Engineering Division, Zaventum, Belgium}
\orcid{0000-0002-9149-7440}

\author{Lucas Visscher}
\affiliation{Department of Chemistry and Pharmaceutical Sciences, Vrije Universiteit, Amsterdam, The Netherlands}
\orcid{0000-0002-7748-6243}

\author{Kareljan Schoutens}
\affiliation{Institute of Theoretical Physics, University of Amsterdam, Amsterdam, The Netherlands}
\affiliation{QuSoft, CWI, Amsterdam, The Netherlands}
\orcid{0000-0002-4159-877X}

\author{Stacey Jeffery}
\affiliation{QuSoft, CWI, Amsterdam, The Netherlands}
\affiliation{KdVI, University of Amsterdam, Amsterdam, The Netherlands}
\orcid{0000-0003-0046-5089}

\maketitle

\begin{abstract}
Near a detailed-balance equilibrium, the perturbed mass-action dynamics of a chemical reaction network (CRN) map exactly onto an electrical-flow problem on the bipartite species–reaction graph: chemical potentials become electrical potentials, Onsager coefficients become conductances, and the instantaneous Gibbs free-energy consumption equals the dissipated electrical energy. We exploit this map to design quantum walk algorithms that decide species reachability, sample reachable species, approximate any individual steady-state reaction flux, and estimate the total Gibbs dissipation. The first three follow from standard electrical-flow quantum walks; the last is non-trivial because the chemical flow is not the minimum-energy electrical flow on the same graph. We resolve this via a new use of alternative neighbourhoods in multidimensional quantum walks, which forces the walker onto the mass-action flow whenever the network is $\sigma$-$M$ rigid. In an adjacency-matrix QRAM access model the algorithms achieve up to a quadratic speedup over classical methods - for example $\Omega(n^{3/2})$ vs $\Omega(n^2)$ for reachability - and dissipation-aware bounds tighten this further when the perturbation is concentrated.
\end{abstract}

\section{Introduction}

Chemical reaction networks (CRNs)~\cite{Feinberg2019foundations}, models of interacting chemical species connected by reaction rules, are central to catalysis~\cite{Steiner2022autonomous}, atmospheric chemistry~\cite{Yang2024CRNexoplanet}, and biochemistry~\cite{Loskot2019BCRNs}, where they are used to reveal short-lived intermediates, elementary reaction mechanisms, and global reaction pathways. Research on CRNs broadly falls into network \emph{generation} and network \emph{analysis}~\cite{Unsleber2020exploration, Turtscher2022pathfinder}: while recent advances have automated the construction of large-scale CRNs~\cite{Stocker2020ML-CRN, Unsleber2022chemoton, Wen2023CRN-ML, Liu2025MLCRNgas}, extracting meaningful structural and dynamical insight from them remains a significant computational challenge.

Perturbations to a CRN can affect its behaviour differently depending on whether its structure is fixed or variable. We focus on fixed-structure CRNs, defined by a finite set of species, e.g.\ $\{A,B,C,D\}$, and a finite set of reactions between these species, e.g.\ $A \rightarrow 2B, A+C \rightarrow D$. Here, the species--reaction graph remains unchanged and perturbations alter only concentrations or external fluxes. Even within this framework, injecting species can shift steady states and activate alternative pathways, increasing the system's effective dimensionality and coupling~\cite{Nicolaou2023prevalence, Banaji2022adding}. Such effects undermine classical algorithms that exploit sparsity or local independence, rendering the analysis of large CRNs computationally challenging despite a fixed underlying network~\cite{Ismail2022graphCRN_review, Turtscher2022pathfinder}. The dynamics on a fixed-structure CRN are most commonly modelled by \textit{mass action kinetics}~\cite{Feinberg2019foundations}, a continuous, deterministic model in which reaction rates depend polynomially on species concentrations; under standard physical assumptions (reversibility, detailed balance, particle conservation) the resulting evolution carries a natural thermodynamic interpretation, in which the system dissipates a free-energy-like function as it relaxes towards a thermodynamic equilibrium. The associated non-linear differential equations are nevertheless analytically intractable and computationally costly at scale~\cite{Doty2018complexityAtCRN, Doty2024complexityDisCRN, Alaniz2023complexitySCRN, Einkemmer2024lowCME}, and many existing graph-based analyses encode them into directed weighted graphs in different and inequivalent ways~\cite{Haag2014interactive, Bergeler2015heuristics, Proppe2016uncertainty, Simm2017contextCRN, Simm2018error, Heuer2018integrated, Simm2018exploration, Proppe2018mechanism, Unsleber2020exploration, Baiardi2021expansive, Steiner2022autonomous, Unsleber2022chemoton, Steiner2024human, Muller2024heron, Csizi2024nanoscale, Weymuth2024scine, Bensberg2024uncertainty, Habershon2015samplingRW, Habershon2016graphrxnpath, Ismail2019auto_multi_rxn_graph, Ismail2022success_challenge_ML, Robertson2019fast_screen_graph, Robertson2020dense_network_rxn, Fakhoury2023contactmap_directedwalks_protein_1, Fakhoury2024contactmap_directedwalks_protein_2, Margraf2019systematic} (see \app{crn-models}).

A particularly clean encoding of mass-action dynamics, however, is the long-standing analogy between CRNs and \emph{electrical networks}. Already in the 1970s, Perelson and Oster~\cite{Perelson1974CRNcircuit} observed that, near a detailed-balance equilibrium, a perturbed mass-action system maps onto a linear flow problem on a weighted graph: species become nodes carrying chemical potentials, reactions become resistive edges whose conductances are determined by the Onsager coefficients, and an external species injection acts as an injected current. The electrical current along an edge then equals the net steady-state reaction flux, and the dissipated electrical energy equals the instantaneous Gibbs free-energy consumption. This correspondence has recently been revisited and sharpened within the framework of stochastic thermodynamics of CRNs~\cite{Avanzini2023circuittheoryCRN}; we emphasise that the analogy itself is not a contribution of this work. What is new here is its algorithmic exploitation: this electrical-flow representation turns out to be precisely the structure that quantum-walk algorithms are designed to use.

Quantum walks based on electrical-network parameters~\cite{belovs2013quantum, apers2022elfs} solve flow-related problems on weighted graphs with up to quadratic speedups over classical random walks, with cost governed by the effective resistance and the total network weight. By instantiating these algorithms on a bipartite species--reaction graph, which we call the \emph{mass-action system graph} (MASG), we obtain quantum algorithms that decide species reachability, sample reachable species, approximate any individual steady-state reaction flux, and estimate the total Gibbs free-energy consumption $\Phi(c)$ of the perturbed system. The first three results follow from off-the-shelf electrical-flow quantum walks. The fourth is technically more subtle: chemical flux on the MASG is constrained by the local stoichiometric ratios at each reaction vertex and consequently does \emph{not} coincide with the minimum-energy electrical flow on the same graph. We resolve this with a new use of \emph{alternative neighbourhoods} from the multidimensional quantum-walk framework~\cite{jeffery2023multidimensional, jeffery2025multidimensional, li2025multidimensional}: whereas alternative neighbourhoods were originally introduced to ease state preparation, we use them to \emph{encode local linear constraints on the flow}, forcing the quantum walker onto the chemically correct flow whenever the network satisfies a graph-theoretic condition we call $\sigma$-$M$ rigidity. To our knowledge, this is the first use of alternative neighbourhoods to sample from a constrained, uniquely-determined flow rather than from the minimum-energy flow, and is of independent interest for the multidimensional quantum-walk programme.. 

\section{Main results}

Assuming a set of natural and physically meaningful constraints on the mass-action system (reversibility, a positive detailed-balance equilibrium $c^*$, and particle conservation; see \sect{thermo}), and given access through QRAM to all thermodynamic quantities at $c^*$, we develop quantum algorithms that probe both structural and dynamical properties of a perturbed CRN. Given a species injection $\eta$ that drives the system to a new steady state $c$, the algorithms:
\begin{enumerate}
    \item[(i)] decide whether a given set of target species $M$ is reachable under $\eta$;
    \item[(ii)] return a representative target species, conditioned on reachability;
    \item[(iii)] approximate the net steady-state flux $J_r(c)$ through any reaction $r$ in the CRN; and
    \item[(iv)] approximate the total Gibbs free-energy consumption $\Phi(c)$.
\end{enumerate}
The algorithms run on a bipartite species--reaction graph that we introduce in \sect{masg}, the \emph{mass-action system graph} (MASG). Their cost is governed by the dissipation $\Phi(c)$ and the total MASG weight ${\sf W}$, and yields up to a quadratic speedup over classical random-walk-based methods on the same graph, sharpening further when the perturbation is concentrated.

From a purely quantum algorithmic perspective, results (iii) and (iv) require more than off-the-shelf machinery: the chemical flow on the MASG does \emph{not} coincide with the minimum-energy electrical flow on the same graph, so standard effective-resistance estimation~\cite{apers2022elfs} only yields a generally loose lower bound on $\Phi(c)$. We close this gap by exhibiting a novel use case of \textit{alternative neighbourhoods} in the multidimensional quantum-walk framework~\cite{jeffery2023multidimensional, jeffery2025multidimensional, li2025multidimensional}: rather than easing state preparation as in their original use, we use them to encode the local stoichiometric ratios at each reaction vertex as constraints on the flow, so that the modified walker is forced onto the chemical flow whenever the network satisfies a graph-theoretic condition we call $\sigma$-$M$ rigidity (\sect{alt}). To our knowledge this is the first use of alternative neighbourhoods to sample from a constrained, uniquely-determined flow rather than the minimum-energy flow.

\section{Preliminaries}\label{sec:prelims}

In this section we collect the graph-theoretic, chemical, and quantum-walk
notions that we use throughout the paper. To keep the main text focused, we
defer pedagogical worked examples and the more technical derivations to
\app{worked-electrical}, \app{linear-response}, and \app{phase-estimation}.

\subsection{Graph theory and electrical networks}\label{sec:graph-elec}

\begin{definition}[Network]\label{def:network}
A network is a connected weighted graph $G = (V,E,\w)$ with a vertex set $V$, an (undirected) edge set $E$ and some weight function $\w:E\rightarrow\mathbb{R}_{>0}$. Since edges are undirected, we can equivalently describe the edges by some set $\E$ such that for all $\{u,v\}\in E$, exactly one of $(u,v)$ or $(v,u)$ is in $\E$. The choice of edge directions is arbitrary. Then we can view the weights as a function $\w:\E\rightarrow\mathbb{R}_{> 0}$, and for all $(u,v)\in \E$, define $\w_{v,u}=\w_{u,v}$. For convenience, we define $\w_{u,v}=0$ for every pair of vertices such that $(u,v)\not\in E$. We write
$${\sf W} := \sum_{(u,v)\in\E}\w_{u,v},$$
for the total weight of the network.

\noindent For an implicit network $G$, and $u\in V$, we will let $\N(u)$ denote the \emph{neighbourhood} of $u$:
$$\N(u):=\{v\in V:\{u,v\}\in E\}.$$
We use the following notation for \emph{the out- and in-neighbourhoods} of $u\in V$:
\begin{equation}
\begin{split}
\N^+(u) &:= \{v\in\N(u):(u,v)\in \E\}\\
\N^-(u) &:= \{v\in\N(u):(v,u)\in \E\},
\end{split}\label{eq:neighbourhoods}
\end{equation}
\end{definition}

To build intuition from physics and apply results from electrical network theory, it is useful to interpret our networks as \textit{electrical networks}.

\begin{definition}[Electrical network]\label{def:elec-network}
Given a network $G = (V,E,\w)$ with a weight function $\w$, we can interpret every edge $\{u,v\}\in E$ as a resistor with resistance $1/\w_{u,v}$. This allows $G$ to be modeled as an \emph{electrical network}.
\end{definition}

\begin{definition}[Flow]\label{def:flow}
A \emph{flow} on a network $G = (V, E, w)$ is a real-valued function on the edges $\f: E \to \mathbb{R}$, such that $\f_{u,v} = -\f_{v,u}$ for every $\{u,v\} \in E$. For any flow $f$ on $G$ and any vertex $u \in V$, we define the net flow leaving $u$ as $\f_u = \sum_{v \in \N(u)} \f_{u,v}$. Flow is said to be \emph{conserved} at $u$ if $\f_u = 0$. A vertex $u$ is called a \emph{source} if $\f_u > 0$ and a \emph{sink} if $\f_u < 0$.

Given a \emph{marked set} $M \subset V$ and an \emph{initial probability distribution} $\sigma$ on $V \setminus M$, a \emph{(unit) \ssm flow} is a flow $\f$ such that each source $u$ is in the support of $\sigma$ and satisfies $\f_u = \sigma(u)$, and each sink is in $M$ with $\sum_{u \in M} \f_u = -1$.

The \emph{energy} of a flow $\f$ is defined as:
\begin{equation*}
{\sf E}(\f):=\sum_{(u,v)\in\E}\frac{\f_{u,v}^2}{\w_{u,v}}.
\end{equation*}

The \emph{effective resistance} ${\sf R}_{\sigma,M}$ is the minimal energy ${\sf E}(\f)$ over all unit \ssm flows $\f$. When $\sigma$ is supported on a single vertex $s$, $M$ is supported on a single vertex $t$, or both, we simplify the notation to ${\sf R}_{s,M}$, ${\sf R}_{\sigma,t}$, or ${\sf R}_{s,t}$, respectively. The \emph{\ssm electrical flow} is the unique unit \ssm flow that achieves this minimal energy.
\end{definition}

\begin{definition}[Potential]\label{def:potential}
A \emph{potential vector} (also known as potential function) on a network $G = (V,E,\w)$ is a real-valued function $\p:V \rightarrow \mathbb{R}$ that assigns a potential $\p_u$ to each vertex $u \in V$.
\end{definition}

Two fundamental laws governing electrical networks are \textit{Kirchhoff's Law} (also known as Kirchhoff's Node Law) and \textit{Ohm's Law}. Kirchhoff's Law defines an \ssm flow as follows:

\begin{definition}[Kirchhoff's Law]\label{def:kcl}
For any given \ssm flow $\f$ on an electrical network $G = (V,E,\w)$, the amount of electrical flow entering any vertex $u \in V \backslash \{\textrm{supp}(\sigma) \cup M\}$ must equal the amount of flow exiting $u$. In other words:
$$ \sum_{v \in \N(u)} \f_{u,v} = 0. $$
\end{definition}

Ohm's Law, on the other hand, states that if a unit of current is injected according to the initial probability distribution $\sigma$ and extracted at the sinks in $M$ of the electrical network $G$, then an induced potential vector $\p$ is generated, as described in \defin{potential}, which is related to the \ssm electrical flow $\f$ in the following manner:

\begin{definition}[Ohm's Law]\label{def:ohm}
    Let $\f$ be the \ssm electrical flow on an electrical network $G = (V,E,\w)$. Then there exists a potential vector $\p$ such that the potential difference between the two endpoints of any edge $\{u,v\} \in E$ is equal to the amount of electrical flow $\f_{u,v}$ along this edge multiplied with the resistance $1/\w_{u,v}$, that is, $\p_u-\p_v=\f_{u,v}/\w_{u,v}$.
\end{definition}

The potential $\p$ induced by an \ssm electrical flow $\f$ in Ohm's Law is not unique. Therefore, it is common practice to consider the potential $\p$ that assigns $\p_u = 0$ for every $u \in M$, in which case $\p_u = \sigma(u){\sf R}_{u,M}$ for every $u \in \textrm{supp}(\sigma)$, where ${\sf R}_{u,M}$ is the effective resistance between $u$ and $M$. A small worked example illustrating these definitions on a four-vertex network is given in \app{worked-electrical}.

\subsection{Chemical reaction networks}\label{sec:crn}

A chemical reaction network (CRN) provides a mathematical framework for representing and analysing a system of complex chemical reactions. It is particularly useful for studying the kinetics of large reaction networks. In this section, we adopt the definition of a CRN, as well as related notions, as introduced by Feinberg~\cite{Feinberg2019foundations}.

\begin{definition}[Chemical Reaction Network]\label{def:crn}
A \textit{chemical reaction network} (CRN) is a tuple \(({\cal S}, {\cal C}, {\cal R})\), where:
\begin{enumerate}
    \item[(i)] \({\cal S}\) is a finite set, whose elements are called the \textit{species} of the network.
    \item[(ii)] \({\cal C} \subset \mathbb{R}_{\geq 0}^{{\cal S}}\) is a finite set of vectors called the \textit{complexes} of the network.
    \item[(iii)] \({\cal R} \subset {\cal C} \times {\cal C}\) is a finite set of ordered pairs called \textbf{reactions}, satisfying the following conditions:
    \begin{enumerate}
        \item[(a)] \(\forall y \in {\cal C}\), \((y, y) \notin {\cal R}\) (no trivial reactions).
        \item[(b)] \(\forall y \in {\cal C}\), there exists \(y' \in {\cal C}\) such that either \((y, y') \in {\cal R}\) or \((y', y) \in {\cal R}\) (every complex participates in at least one reaction either as a reactant or as a product.
    \end{enumerate}
\end{enumerate}
For a reaction \((y, y') \in {\cal R}\), we denote it by \(y \to y'\), where \(y\) is called the \textit{reactant complex} and \(y'\) is the \textit{product complex}.
\end{definition}

As an example, consider the CRN defined by
\begin{align}\label{eq:example-crn}
    {\cal S} = \{&A, B, C, D, E\},\\
    {\cal C} = \{&A, 2B, A + C, D, B + E\},\\
    {\cal R} = \bigl\{&A \to 2B, 2B \to A, A + C \to D, \nonumber\\
    &\quad D \to A + C, D \to B + E,  \\
    &\quad B + E \to A + C \nonumber \bigr\}.
\end{align}

This corresponds to the reaction system whose standard reaction diagram is displayed as
\[
\begin{aligned}
    A &\rightleftharpoons 2B, \\
    A + C &\rightleftharpoons D, \\
    D &\rightarrow B + E, \\
    B + E &\rightarrow A + C.
\end{aligned}
\]

A CRN is not only characterized by its static network structure but also by the time evolution of species concentrations. These dynamics are governed by the \emph{kinetics} of the CRN. The kinetics act on a \emph{concentrations vector} \(c \in \mathbb{R}_{\ge 0}^{{\cal S}}\), whose components \(c_s\) (for each \(s \in {\cal S}\)) specify the concentration of \(s\)-molecules per unit volume of the mixture. In this work, we consider the following specific form of kinetics:

\begin{definition}[Mass Action Kinetics]\label{def:kin}
A \emph{mass action kinetics} for a chemical reaction network
\(({\cal S}, {\cal C}, {\cal R})\)
is an assignment to each reaction \(y \to y'\) of a \emph{mass action rate function}
\({\cal K}_{y \to y'} \colon \mathbb{R}_{\ge 0}^{{\cal S}} \to \mathbb{R}_{\ge 0}\)
such that there exists a positive number \(k_{y \to y'}\) with
\[
{\cal K}_{y \to y'}(c) = k_{y \to y'} \prod_{s \in {\cal S}} c_{s}^{y_s}.
\]
Here \(k_{y \to y'}\) is the \emph{rate constant} for the reaction \(y \to y'\), and the component \(y_s\) of \(y\) is the \emph{stoichiometric coefficient} of species \(s\) in the reactant complex \(y\).
\end{definition}

In the above definition we adopt the convention \(0^0 = 1\) when a particular concentration \(c_s\) is zero, and abbreviate $\prod_{s \in {\cal S}} c_s^{y_s} = c^y$, so that ${\cal K}_{y \to y'}(c) = k_{y \to y'} c^y$. A worked computation of the rate functions for two of the reactions in \eq{example-crn} is given in \app{rate-example}. A mass action kinetics for a CRN $({\cal S}, {\cal C}, {\cal R})$ is completely specified by assigning a positive rate constant $k_{y \to y'}$ to each reaction $y \to y' \in {\cal R}$, that is, by an element $k \in \mathbb{R}_+^{\cal R}$. We therefore speak of a \textit{mass action system (MAS)} $({\cal S}, {\cal C}, {\cal R},k)$, which is the CRN $({\cal S}, {\cal C}, {\cal R})$ combined with the mass action kinetics uniquely determined by~$k$.

\subsection{The thermodynamics of an MAS}\label{sec:thermo}

In this section, we delve deeper into the thermodynamics of an MAS, following the notions in~\cite{de2013non}. When studying the kinetics of an MAS, it is convenient to study \textit{reversible} reactions. For a CRN $({\cal S}, {\cal C}, {\cal R})$, we say that a reaction \(y \to y' \in {\cal R} \) is reversible if \(y' \to y \in {\cal R}\). For a reversible reaction \(y \to y'\), we define its \textit{net flux} $J_{y \to y'}(c)$ as
\begin{equation}\label{eq:netflux}
    J_{y \to y'}(c) := {\cal K}_{y \to y'}(c) - {\cal K}_{y' \to y}(c).
\end{equation}
This net flux is positive if the reaction proceeds predominantly in the forward direction, negative if the backward direction is dominant, and zero when forward and backward processes are perfectly balanced.

If all reactions in ${\cal R}$ for a CRN $({\cal S}, {\cal C}, {\cal R})$ are reversible, we say that the CRN is reversible. To avoid listing both directions of each reversible pair of reactions, we fix a subset $\R \subset {\cal R}$ called the set of \textit{oriented reactions}, containing exactly one direction from every reversible pair $y \to y', y' \to y \in {\cal R}$, similarly to how the directed edge set $\E$ is obtained from the edge set $E$ in a graph. We can fully characterise the set of reversible reactions ${\cal R}$ from just $\R$. If $r = y \to y'$, we use $\overline{r}$ to denote $y' \to y$. Depending on the chosen direction, $J_r(c)$ would change sign.

For a reversible CRN, if the net flux of all reactions is zero, that is, there exists a vector of concentrations \(c^* \in \mathbb{R}_{\ge 0}^{\cal S}\) such that, for each reaction $r \in \R$,
\begin{equation}\label{eq:db}
    {\cal K}_{r}(c^*) = {\cal K}_{\overline{r}}(c^*),
\end{equation}
then we say that the MAS is in \textit{detailed balance}. Such a vector \(c^*\) is often referred to as a (thermodynamic) \emph{equilibrium} of the MAS.

One may then ask what happens if we perturb this equilibrium by introducing or removing certain amounts of one or more species. Such a perturbation generally alters the exact cancellation of forward and backward fluxes, typically driving the system away from detailed balance. For a reaction $r = y \to y'$, let $\nu_{r,s} := y'_s - y_s$ denote the \textit{net stoichiometric coefficient} of $s$ in $r$. Then the system, with new concentrations $c$, will still satisfy the \textit{steady-state} condition for each species:
\begin{equation}\label{eq:steady}
    \frac{d c_s}{dt} = \sum_{r \in \R} \nu_{r,s}J_{r}(c) = \eta_s,
    \quad
    \text{for all species } s \in {\cal S}.
\end{equation}
Here \(\eta_s \in\mathbb{R} \) denotes the net \emph{external injection/removal rate} of species \(s\). By convention, \(\eta_s>0\) means species \(s\) is injected from an external source, while \(\eta_s<0\) indicates a net removal of \(s\). Note that \eq{steady} is invariant under the choice of direction in $\R$, since
\[\nu_{r,s}J_{r}(c) = -\nu_{\overline{r},s}J_{\overline{r}}(c).\]

We now introduce \textit{chemical potentials} $\{\mu_s(c)\}_{s \in {\cal S}}$. For a mass action system, these are defined as
\begin{equation}\label{eq:chempot}
    \mu_s(c) := RT \ln(c_s) + \mu_s^0,
\end{equation}
where $R$ and $T$ are the gas constant and absolute temperature, respectively, and $\mu_s^0$ is a reference potential which we may ignore as we are only interested in changes of $\mu_s(c)$ under a perturbation of concentrations. A standard linear-response argument (see \app{linear-response} for the full derivation) yields, in the limit \(\delta c_s/c^{*}_s\to 0\),
\begin{equation}\label{eq:chempot-change}
  \delta \mu_s(c) = RT\,\frac{\delta c_s}{c^{*}_s},
\end{equation}
and applying the same first-order expansion to $J_{r}(c)$ around $c^*$ gives
\begin{equation}\label{eq:netflux-change}
    \delta J_{r}(c) \;=\; -{\cal K}_{r}(c^*)\sum_{s \in {\cal S}} \frac{\nu_{r,s}}{c^{*}_s} \,\delta c_s.
\end{equation}
Introducing the \emph{affinity} of the reaction $r$,
\[
\Delta\mu_{r}(c) := -\sum_{s \in {\cal S}}\nu_{r,s} \,\delta\mu_s(c),
\]
and the \emph{Onsager coefficient}
\[
G_{r} := \frac{{\cal K}_{r}(c^*)}{RT}>0,
\]
we obtain from \eq{netflux-change} and \eq{db} the linear flux--affinity relation
\begin{equation}\label{eq:netflux-change2}
    J_{r}(c) = J_{r}(c^*) + \delta J_{r} = G_{r}\,\Delta\mu_{r}(c).
\end{equation}

The last thermodynamic quantity we need is the \emph{instantaneous rate of Gibbs free-energy consumption} $\Phi$:
\begin{equation*}
\Phi(c) = \sum_{r \in \R} J_r(c)\,\Delta\mu_r(c),
\end{equation*}
which measures the free-energy dissipated, i.e.\ the chemical ``power'' irreversibly lost (or supplied) by the running reactions, per unit time. Near equilibrium, by \eq{netflux-change2},
\begin{equation}\label{eq:free-energy}
    \Phi(c) = \sum_{r \in \R} J_r(c)^2/G_r.
\end{equation}
Due to the square, $\Phi(c)$ is invariant under the choice of direction in $\R$, since $G_r = G_{\overline{r}}$.

\subsection{Quantum walks}\label{sec:qw}

We employ quantum walk algorithms to compute properties of our chemical reaction networks. In contrast to classical random walks, which explore a graph diffusively, quantum walks exhibit interference and coherence and thereby admit quadratic speed-ups in search-type tasks~\cite{szegedy2004QMarkovChainSearch, magniez2007search, ambainis2019QuadSpeedupFindingMarkedQW, apers2019UnifiedFrameworkQWSearch} and, in some cases, even exponential improvements for structured graph problems~\cite{childs2003ExpSpeedupQW}. These features suggest that quantum algorithms could help overcome the combinatorial complexity inherent in CRNs, where the number of pathways, intermediates, and reaction configurations grows rapidly with the system size.

We adopt the electrical-network framework introduced by~\cite{belovs2013quantum} and further refined by~\cite{piddock2019quantum, apers2022elfs}; equivalent algorithms can be derived via the span-program formalism~\cite{ito2019approximate, jeffery2023quantum}. Throughout this work we treat this framework as a black box that produces quantum walk algorithms with specified problem-solving capabilities and complexities, deferring the underlying phase-estimation construction to \app{phase-estimation}.

The algorithms operate on the \emph{edge space} of $G$,
\begin{equation}\label{eq:hilbert}
    {\cal H}=\mathrm{span}\{\ket{u,v}:\{u,v\} \in E(G)\},
\end{equation}
in which each undirected edge $\{u,v\}$ appears twice, as $\ket{u,v}$ and $\ket{v,u}$. Two states play a distinguished role. The first is the normalised \emph{star state} of a vertex $u\in V(G)$,
\begin{equation}\label{eq:star-state}
    \ket{\psi_\star(u)} :=\frac{1}{\sqrt{\w_u}}\sum_{v\in\N(u)}\sqrt{\w_{u,v}}\ket{u,v},
\end{equation}
where $\w_u := \sum_{v \in \N(u)} \w_{u,v}$ is the \textit{weighted degree} of $u$. Measuring $\ket{\psi_\star(u)}$ yields the edge $\ket{u,v}$ with probability $\w_{u,v}/\w_u$, so $\ket{\psi_\star(u)}$ is a quantum analogue of the local classical-walk transition. The second is the (normalised) \emph{flow state} associated with the \sm electrical flow $\f$,
\begin{equation}\label{eq:flowstate}
    \ket{\f} := \frac{1}{\sqrt{2{\sf E}(\f)}}\sum_{(u,v) \in \E(G)}\frac{\f_{u,v}}{\sqrt{\w_{u,v}}}\left(\ket{u,v} - \ket{v,u}\right),
\end{equation}
which lives in the antisymmetric subspace of ${\cal H}$.

\subsubsection*{Quantum walk algorithms}\label{sec:walk-algo}

Depending on the application, we employ one of four quantum walk algorithms, each based on phase estimation of a unitary $U_{{\cal A}{\cal B}} = (2\Pi_{\cal A} - I)(2\Pi_{\cal B} - I)$ acting on ${\cal H}$. Here, $\Pi_{\cal A}$ denotes the projection onto the subspace ${\cal A}$ and hence $(2\Pi_{\cal A} - I)$ is the reflection through ${\cal A}$. The full construction of the subspaces ${\cal A}, {\cal B}$, the initial state $\ket{\psi_0}$, and the operator $U_{{\cal A}{\cal B}}$ is given in \app{phase-estimation}.

The complexities of these algorithms depend on the required precision of the phase estimation, the cost of generating $\ket{\psi_0}$, and the cost of applying $U_{{\cal A}{\cal B}}$. Since the initialisation of $\ket{\psi_0}$ and the reflections around ${\cal A}, {\cal B}$ differ depending on the chosen algorithm, we express their complexities using two abstract complexity parameters. For a graph $G$, marked set $M \subset V$ and distribution $\sigma$, define:
\begin{itemize}
    \item ${\sf S}$: an upper bound on the cost of preparing the quantum state
    \begin{equation}\label{eq:sigma}
       \ket{\sigma} = \sum_{u \in \text{supp}(\sigma)} \sigma(u) \ket{u}.
    \end{equation}
    This parameter captures the complexity of generating the initial state $\ket{\psi_0}$ (for further details, see Section~3.2.8 in~\cite{jeffery2023multidimensional}).
    \item ${\sf U}_{\star}$: an upper bound on the cost of verifying whether a vertex $u$ belongs to $\text{supp}(\sigma)$, the marked set $M$, or neither, as well as the cost of generating the star state $\ket{\psi_{\star}(u)}$ (see \eq{star-state}) for any $u \in V \setminus (\{s\} \cup M)$. This parameter captures the complexity of implementing the quantum walk operator $U_{{\cal AB}}$ (for further details, see Section~3.2.5 in~\cite{jeffery2023multidimensional}).
\end{itemize}

\noindent The first algorithm detects whether the marked set $M$ is non-empty:

\begin{theorem}[\cite{belovs2013quantum}]\label{thm:detect}
    Fix a network $G = (V,E,\w)$ as in \defin{network}, a marked set $M \subset V$ and an initial probability distribution $\sigma$ on $V \setminus M$, with the promise that either $M$ is empty or there is a path in $G$ connecting $\sigma$ to $M$. Then there exists a quantum walk algorithm that decides whether $M$ is empty with a constant probability of success in cost
    $$O\!\left({\sf S} + \sqrt{{\sf R}_{\sigma,M}{\sf W}}\,{\sf U}_{\star}\right).$$
\end{theorem}

This is the only algorithm in this work that is ``purely'' a phase-estimation algorithm; the remaining ones combine (a slightly differently initialised) phase-estimation routine with other well-known quantum subroutines. We refer the interested reader to the original sources for details. Our second algorithm returns a marked vertex rather than only deciding existence:

\begin{theorem}[\cite{piddock2019quantum}]\label{thm:find}
    Fix a network $G = (V,E,\w)$ as in \defin{network}, a non-empty marked set $M \subset V$ and an initial probability distribution $\sigma$ on $V \setminus M$, with the promise that there is a path in $G$ connecting $\sigma$ to $M$. Then there exists a quantum walk algorithm that returns a marked element from $M$ with a constant probability of success in cost
    $$O\!\left({\sf S} + \sqrt{{\sf R}_{\sigma,M}{\sf W}}\log^3(|M|)\,{\sf U}_{\star}\right).$$
\end{theorem}

The next two algorithms consider the special case where $\sigma$ is supported on a single vertex and return information about the electrical network. Their complexities involve a quantity from~\cite{apers2022elfs} called the \emph{escape time} ${\sf ET}_{s,M}$:
\begin{equation}\label{eq:escape}
    {\sf ET}_{s,M} := \frac{1}{{\sf R}_{s,M}} \sum_{u\in V(G)}\p_u^2\w_u.
\end{equation}
Operationally, ${\sf ET}_{s,M}$ is the expected time at which a random walk leaves $s$ for the final time before arriving in $M$. Since the effective resistance forms a metric, $\p_u \le \p_s = {\sf R}_{s,M}$ for any $u \in V$, so ${\sf ET}_{s,M}$ lower-bounds ${\sf R}_{s,M}{\sf W}$. The third and fourth algorithms below admit generalisations via the \emph{alternative neighbourhood} technique of~\cite{jeffery2023multidimensional}, which we discuss in \sect{gibbs}. The first of the two estimates the effective resistance ${\sf R}_{s,M}$ (up to the weighted degree $\w_s$ of $s$).

\begin{theorem}[\cite{apers2022elfs}]\label{thm:effective}
    Fix a network $G = (V,E,\w)$ as in \defin{network}, a non-empty marked set $M \subset V$ and an initial vertex $s \in V \setminus M$, with the promise that there is a path in $G$ connecting $s$ to $M$. Then there exists a quantum walk algorithm that $\epsilon$-multiplicatively estimates ${\sf R}_{s,M}\w_s$ with a constant probability of success in cost
    $$O\!\left(\frac{1}{\epsilon}\left({\sf S} + \frac{1}{\epsilon}\left({\sf ET}_{s,M} + \log({\sf R}_{s,M}\w_s)\right){\sf U}_{\star}\right)\right).$$
\end{theorem}

The second generates an approximation of the flow state corresponding to the \sm electrical flow $\f$:

\begin{theorem}[\cite{apers2022elfs,jeffery2023quantum}]\label{thm:flow}
    Fix a network $G = (V,E,\w)$ as in \defin{network}, a non-empty marked set $M \subset V$ and an initial vertex $s \in V \setminus M$, with the promise that there is a path in $G$ connecting $s$ to $M$. Let $\f$ be the \sm electrical flow on $G$ with corresponding flow state $\ket{\f}$ as defined in \eq{flowstate}. Then there exists a quantum walk algorithm that returns a state $\ket{\widetilde{\f}}$ satisfying
    $$ \frac{1}{2}\norm{\proj{\widetilde{\f}} - \proj{\f}} \leq \epsilon$$
    with a constant probability of success in cost
    $$O\!\left({\sf S} + \frac{1}{\epsilon^2}\left(\sqrt{{\sf ET}_{s,M}} + \log({\sf R}_{s,M}\w_s)\right){\sf U}_{\star}\right).$$
\end{theorem}

To run these algorithms we need bounds on ${\sf R}_{s,M}$, $\w_s$, ${\sf W}$ and ${\sf ET}_{s,M}$ in order to set the precision of the phase-estimation subroutine. If these are not known beforehand, the algorithms (and their corresponding complexities) can be instantiated with (trivial) upper bounds on these quantities.

\section{Representing an MAS as an electrical network}\label{sec:masg}

\subsection{The mass action system graph}

In this work, we represent an MAS as an electrical network, to allow us to run a quantum walk algorithm on this network. To create this electrical network, we base our approach on the bipartite graph representation by~\cite{Sakamoto1988GraphCRN}. In this bipartite graph representation, the set of species, \({\cal S}\), forms one partition of the vertices and the set of reactions, \({\cal R}\), forms the other. A directed edge, or \textit{arc}, is drawn between a species \(s \in {\cal S}\) and a reaction \(y \to y' \in {\cal R}\) if and only if species \(s\) occurs with non-zero stoichiometry in the reactant complex $y$, that is, \(y_s \neq 0\). Similarly, an arc is drawn between a reaction \(y \to y' \in {\cal R}\) and a species \(s \in {\cal S}\) when \(s\) occurs with non-zero stoichiometry in the product complex $y'$. Note that although the formal definition of a chemical reaction network (CRN) includes the set of complexes \( {\cal C} \), this representation does not. In many analyses, including our graph-based approach, the essential information provided by the complexes is fully encoded within the reactions themselves. Each reaction inherently describes the transformation from a set of reactant species to a set of product species. Therefore, by representing the CRN as a bipartite graph with one partition corresponding to the species \( {\cal S} \) and the other to the reactions \( {\cal R} \), we retain all critical stoichiometric information. Unlike other graph representations of CRNs, this bipartite graph ensures that each reactant is represented only once while remaining correctly connected to its associated reactions.

To represent an MAS as an electrical network, we are not allowed to have any directed edges (see \defin{network}), meaning that this approach will not suffice. In addition, the thermodynamic machinery discussed in \sect{thermo} also needs our MAS to have additional structure. We will therefore assume in the rest of this work that our MAS \(({\cal S}, {\cal C}, {\cal R}, k)\) satisfies the following three conditions:
\begin{enumerate}
\item The underlying CRN of our MAS is \textbf{reversible} (every reaction $y\to y'$ appears together with its backward reaction $y'\to y $ in ${\cal R}$), meaning we can fully characterise ${\cal R}$ by the set of oriented reactions $\R$.
\item The MAS admits a positive equilibrium concentrations vector $c^*$ that realises \textbf{detailed balance}, i.e.\ the forward and reverse fluxes of every reaction pair coincide at $c^*$.
\item Each reaction in ${\cal R}$ is \textbf{particle conserving}, meaning that for every $y \to y'\in {\cal R}$ we have $\sum_s y_s = \sum_s y'_s$.
\end{enumerate}
In practice, many mechanistic models violate at least one of these assumptions: combustion schemes contain intrinsically irreversible steps, metabolic networks include ``committed'' reactions, and numerous biological or catalytic processes are deliberately driven far from equilibrium, so that neither a detailed-balance concentration nor particle conservation holds exactly. Nevertheless, imposing these constraints is a standard assumption to simplify the analysis of chemical kinetics within the CRN framework. Additionally, we shall shortly see that under these assumptions, the dynamical properties of our MAS will give rise to a duality with electrical networks.

Under the assumptions of reversibility, detailed balance, and particle conservation, our bipartite graph representation is formally defined as follows:

\begin{definition}[Mass Action System Graph]\label{def:masg}
For an MAS \(({\cal S}, {\cal C}, {\cal R}, k)\) and a concentrations vector $c \in \mathbb{R}_{\geq 0}^{\cal S}$ its \textit{mass action system graph} (MASG) is a weighted bipartite undirected graph \(G = (V, \overline{E}, \w)\), where:
\begin{itemize}
    \item The vertex set \(V\) is the disjoint union of species and (reversible) reactions, given by \(V = {\cal S} \sqcup \R\), where \({\cal S}\) is the set of species and \(\R\) is the set of oriented reactions.
    \item The edge set \(E \subseteq \{\{s, r\} \mid s \in {\cal S}, r \in \R\}\) consists of edges connecting species to reactions. An edge \(\{s, r\}\) is part of $E$ if and only if $s$ appears with non-zero stoichiometry in $r$, i.e.\ at least one of $y_s,y'_s$ is non zero for $r = y \to y'$. The directed edge set $\E$ is obtained by considering the direction $(s,r)$ for every $\{s,r\} \in E$.
    \item The weight function \(\w: E \to \mathbb{R}_{ > 0}\) is defined on the edge set \(E\), where each edge $\{s, r\}$ is weighted by $\nu_r\abs{\nu_{r,s}}G_r$, where $\nu_r := \sum_{s \in {\cal S}}\abs{\nu_{r,s}}$.
\end{itemize}
\end{definition}

Note that $G$ is invariant under the choice of direction in $\R$, since $\nu_r = \nu_{\overline{r}}$. To obtain some intuition for the construction of MASG in \defin{masg}, we return to our example from \eq{example-crn}. We make some small modifications to the CRN to ensure that it is reversible and particle conserving. For readability, we also relabel the reactions as follows:
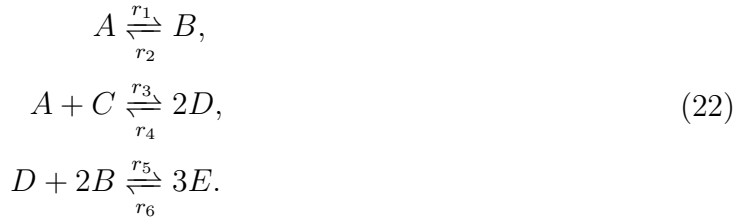
\begin{equation}\label{eq:example-balance}
\begin{split}
    A &\xrightleftharpoons[r_2]{r_1} B, \\
    A + C &\xrightleftharpoons[r_4]{r_3} 2D, \\
    D + 2B &\xrightleftharpoons[r_6]{r_5} 3E.
\end{split}
\end{equation}
The resulting MASG is shown in \fig{example-masg}.

\begin{figure}
\centering
\begin{tikzpicture}
    \node[draw, circle] (R1) at (6,8) {$r_1$};
    \node[draw, circle] (R2) at (6,5) {$r_3$};
    \node[draw, circle] (R3) at (6,2) {$r_5$};

    \node[draw, circle] (A) at (0,9) {A};
    \node[draw, circle] (B) at (0,7) {B};
    \node[draw, circle] (C) at (0,5) {C};
    \node[draw, circle] (D) at (0,3) {D};
    \node[draw, circle] (E) at (0,1) {E};
    \draw[->] (A) to (R1);
    \draw[->] (B) to (R1);
    \draw[->] (A) to (R2);
    \draw[->] (C) to (R2);
    \draw[->] (D) to (R2);
    \draw[->] (D) to (R3);
    \draw[->] (B) to (R3);
    \draw[->] (E) to (R3);
    \node at (1.5,9) {$2G_{r_1}$};
    \node at (1.5,8.4) {$4G_{r_3}$};
    \node at (1.5,7.5) {$2G_{r_1}$};
    \node at (1.5,6.3) {$12G_{r_5}$};
    \node at (1.5,5.25) {$4G_{r_3}$};
    \node at (1.5,3.85) {$8G_{r_3}$};
    \node at (1.5,3) {$6G_{r_5}$};
    \node at (1.5,1.6) {$18G_{r_5}$};
\end{tikzpicture}
\caption{The resulting MASG for the CRN in \eq{example-balance}.}
\label{fig:example-masg}
\end{figure}

\subsection{The flow on an MASG}

For an MAS \(({\cal S}, {\cal C}, {\cal R}, k)\), consider a marked set of species $M \subseteq {\cal S}$ and a probability distribution $\sigma$ on ${\cal S}\setminus M$. Let $c_{\sigma,M} \in \mathbb{R}_{\geq 0}^{\cal S}$ be the concentrations vector obtained by slightly perturbing the equilibrium $c^*$ through the addition of species, characterised by $\eta$ (see \eq{steady}), such that
\begin{align}\label{eq:eta}
    &\eta_s = \sigma(s), &&\sum_{s \in M}\eta_s = -1.
\end{align}
We exhibit a \ssm flow $\f$ on the MASG obtained from the MAS $({\cal S}, {\cal C}, {\cal R}, k)$ and concentrations vector $c_{\sigma,M}$. This flow $\f$ will hence have to satisfy Kirchhoff's Law (see \defin{kcl}), but it does not necessarily coincide with the actual \ssm electrical flow on the MASG, meaning it does not need to satisfy Ohm's Law (see \defin{ohm}).

\begin{lemma}\label{lem:masgflow}
For an MASG \(G = (V, E, \w)\) obtained from an MAS \(({\cal S}, {\cal C}, {\cal R}, k)\) and a concentrations vector $c \in \mathbb{R}_{\geq 0}^{\cal S}$, the following assignment of flow $\f$ to the edges $\E$ is a valid unit \ssm flow on $G$:
\begin{equation}\label{eq:masgflow}
    \f_{s,r} = \nu_{r,s}J_r(c_{\sigma,M}),
\end{equation}
where $J_r(c_{\sigma,M})$ is the net flux of the reaction $r$ under the concentrations vector $c_{\sigma,M}$ (see \eq{netflux}).
\end{lemma}
\begin{proof}
    To show that $\f$ is a unit \ssm flow on $G$, we must show that it satisfies Kirchhoff's Law, as well as that $\f_s = \sigma(s)$ for every $s$ in the support of $\sigma$ and $\sum_{s \in M}\f_s = -1$.

    We split Kirchhoff's Law into two parts tackling the vertices in ${\cal S}$ and $\R$ separately. Starting with $\R$, we know by assumption that our MAS is particle conserving. Hence, for every $r = y \to y' \in \R$, 
    \begin{align*}
    \sum_{s \in {\cal S}}\f_{r,s} &= -J_r(c_{\sigma,M})\sum_{s \in {\cal S}}\nu_{r,s}\\
    &= -J_r(c_{\sigma,M})\sum_{s \in {\cal S}}(y'_s - y_s) = 0.   
    \end{align*}
    
    We combine Kirchhoff's Law for the vertices in ${\cal S}$ with the other requirements, namely $\f_s = \sigma(s)$ for every $s$ in the support of $\sigma$ and $\sum_{s \in M}\f_s = -1$, since all of these follow from the fact that $c_{\sigma,M}$ satisfies the steady-state condition (see \eq{steady}), meaning that for every $s \in {\cal S}$ we have:
    \[
    \sum_{r \in \R} \f_{r,s} = \sum_{r \in \R} \nu_{r,s}J_{r}(c_{\sigma,M}) = \eta_s.
    \]
    By \eq{eta}, we know that $\eta_s = \sigma(s)$ for every $s \in {\cal S} \setminus M$ and that $\sum_{s \in M}\eta_s = -1$, which completes the proof.
\end{proof}

We will refer to the flow from \eq{masgflow} as the MASG flow. Now that we know that the MASG flow is a unit \ssm flow by \lem{masgflow}, we can calculate its energy ${\sf E}(\f)$:
\begin{equation}\label{eq:energy-masg}
\begin{split}
    {\sf E}(\f) &= \sum_{s \in {\cal S}}\sum_{r \in \R}\frac{\f_{s,r}^2}{\w_{s,r}} = \sum_{s \in {\cal S}}\sum_{r \in \R}\frac{\nu_{r,s}^2J_r(c_{\sigma,M})^2}{\nu_r\abs{\nu_{r,s}}G_r}\\
    &= \sum_{r \in \R}\frac{J_r(c_{\sigma,M})^2}{G_r}.
\end{split}
\end{equation}

Note that this precisely matches the instantaneous rate of Gibbs free-energy consumption $\Phi(c_{\sigma,M})$ from \eq{free-energy}.

In this section, we have shown that an MASG can be represented as an electrical network, such that its dynamics map directly to those of the electrical network, and hence to the quantum-walk parameters. \tabl{param} summarises this connection.

\begin{table*}
    \renewcommand{\arraystretch}{1.5}
    \centering
    \begin{tabular}{|c|c|}
        \hline
        MASG & Electrical network\\ \hline
        Species ${\cal S}$ & Vertices ${\cal S}$ forming an independent set \\ \hline
        Oriented reactions $\R$ & Vertices $\R$ forming an independent set \\ \hline
        Target species in ${\cal S}$ & Marked set $M \subseteq {\cal S}$ \\ \hline
        External injection/removal rate $\eta_s$ & Initial probability distribution $\sigma(s)$ \\ \hline
        Non-zero stoichiometric coefficient of $s \in {\cal S}$ in $r \in \R$ & Directed edge $(s,r)$ \\ \hline
        Thermodynamical quantity $\nu_r\abs{\nu_{r,s}}G_r$ & Edge conductance $\w_{s,r}$ \\ \hline
        Thermodynamical quantity $\nu_{r,s}J_r(c_{\sigma,M})$ & Valid unit \sm flow $\f_{s,r}$ \\ \hline
        Gibbs free-energy consumption $\Phi(c_{\sigma,M})$ & Energy ${\sf E}(\f)$ of the MASG flow \\ \hline
    \end{tabular}
    \caption{A summary of how our MASG representation connects the CRN and its dynamics to an electrical network.}\label{tab:param}
\end{table*}

\subsection{The computational cost}\label{sec:cost}

Using our electrical network representation of the MAS, we develop quantum walk algorithms to answer structural and kinetic questions about the underlying chemical reaction network. These constructions apply to MASs that are particle-conserving, reversible, and admit a detailed balance condition.

The asymptotic cost of these algorithms relies on the abstract costs from \sect{walk-algo}: ${\sf S}$, i.e.\ setting up the state $\ket{\sigma}$, and $U_{\star}$, i.e.\ the cost of implementing the quantum walk operator $U_{\cal{AB}}$ from \eq{U-AB}, or equivalently, the cost of implementing the star states $\ket{\psi_\star(u)}$. The states relevant to these costs are therefore
\begin{equation}\label{eq:queries}
\begin{split}
    &\ket{\sigma} = \sum_{s \in {\cal S} \setminus M}\sqrt{\eta_s}\ket{s},\\
    &\ket{\psi_\star(s)} \propto \sum_{r \in \R}\sqrt{\nu_r\abs{\nu_{r,s}}G_r}\ket{s,r},\\
    &\ket{\psi_\star(r)} = \frac{1}{\sqrt{\nu_r}}\sum_{s \in {\cal S}}\sqrt{\nu_{r,s}}\ket{r,s}.
\end{split}
\end{equation}
For typical applications of our framework it is reasonable to assume that these quantum states can be prepared efficiently. Because we focus on near-equilibrium dynamics (see \sect{thermo}), we assume that the MAS is initially close to a known steady state $c^*$ and that all thermodynamic quantities required with respect to this equilibrium can be pre-computed classically. Storing these values in read-only QRAM then allows efficient generation of the states in \eq{queries}.

\subsection{Detecting and finding reachable species}

We are now ready to describe our quantum algorithms for deciding whether a given set of target species is reachable, and for outputting a target species if this is the case.

\begin{corollary}\label{cor:detect-mas}
    Fix an MAS \(({\cal S}, {\cal C}, {\cal R}, k)\) that admits a detailed balance, that is particle conserving and whose underlying CRN is reversible. Fix an initial distribution of species $\sigma$ on ${\cal S}$ and a set of target species $M$, and let $\sf S$ be the cost of generating the state $\ket{\sigma}$. Then there exists a quantum walk algorithm that decides whether any of the target species in $M$ is reachable from $\sigma$ with a constant probability of success in cost
    \begin{align*}
        {\sf S} + \sqrt{\Phi(c_{\sigma,M})\sum_{r \in \R}\nu_r^2G_r}.
    \end{align*}
\end{corollary}
\begin{proof}
    We invoke \thm{detect} to obtain a quantum walk that runs on the MASG obtained from our MAS. By our choice of weights for the MASG, we know that
    $$ {\sf W} = \sum_{r \in \R}\sum_{s \in {\cal S}} \nu_r \abs{\nu_{r,s}} G_r = \sum_{r \in \R}\nu_r^2G_r.$$
    Lastly, we know by \eq{energy-masg} that the energy of the MASG flow allows us to upper bound the effective resistance as follows:
    \begin{equation*} {\sf R}_{\sigma,M} \leq {\sf E}(\f) = \Phi(c_{\sigma,M}). \qedhere\end{equation*}
\end{proof}

\begin{corollary}\label{cor:find-mas}
    Fix an MAS \(({\cal S}, {\cal C}, {\cal R}, k)\) that admits a detailed balance, that is particle conserving and whose underlying CRN is reversible. Fix an initial distribution of species $\sigma$ on ${\cal S}$ and a non-empty set of target species $M$, reachable from $\sigma$. Then there exists a quantum walk algorithm that returns any of the target species from $M$ with a constant probability of success in cost
    \begin{align*}
        &{\sf S} + \sqrt{\Phi(c_{\sigma,M})\sum_{r \in \R}\nu_r^2G_r}\log^3(\abs{M}).
    \end{align*}
\end{corollary}
\begin{proof}
    This follows directly from \thm{find}, by applying the same parameters as in \cor{detect-mas}.
\end{proof}

\subsection{Worked example: a five-species reaction network}\label{sec:example}

We instantiate the framework on the five-species CRN already introduced in \eq{example-balance} and \fig{example-masg}. To make the
algorithmic costs concrete, we set the Onsager coefficients
$G_{r_1} = G_{r_3} = G_{r_5} = 1$ (in any consistent unit system) and
consider the perturbation
\begin{equation}
\sigma(E) = 1, \qquad M = \{A, C\},
\label{eq:example-perturbation}
\end{equation}
i.e.\ a unit injection of species $E$ together with extraction at
species $A$ and $C$. Physically, this is a controlled probe of the
reverse pathway $E \to A,C$.

\textbf{Steady-state fluxes.} Imposing Kirchhoff's law at every species node of the MASG fixes the
mass-action fluxes uniquely:
\begin{equation}
J_{r_5} = \tfrac{1}{3}, \qquad
J_{r_1} = \tfrac{2}{3}, \qquad
J_{r_3} = \tfrac{1}{6},
\label{eq:example-fluxes}
\end{equation}
and the sink decomposition $\eta_A = -\tfrac{5}{6}$, $\eta_C =
-\tfrac{1}{6}$. The corresponding edge flow $\theta_{s,r} = \nu_{r,s}
J_r$ on the MASG is shown in \fig{example-flow}.

\begin{figure}
\centering
\begin{tikzpicture}[scale=0.85]
    \node[draw, circle] (R1) at (6,8) {$r_1$};
    \node[draw, circle] (R2) at (6,5) {$r_3$};
    \node[draw, circle] (R3) at (6,2) {$r_5$};

    \node[draw, circle] (A) at (0,9) {A};
    \node[draw, circle] (B) at (0,7) {B};
    \node[draw, circle] (C) at (0,5) {C};
    \node[draw, circle] (D) at (0,3) {D};
    \node[draw, circle] (E) at (0,1) {E};

    \draw[->] (A) to (R1);
    \draw[->] (B) to (R1);
    \draw[->] (A) to (R2);
    \draw[->] (C) to (R2);
    \draw[->] (D) to (R2);
    \draw[->] (D) to (R3);
    \draw[->] (B) to (R3);
    \draw[->] (E) to (R3);

    \node at (1.5,9.15)    {$-\tfrac{2}{3}$};   
    \node at (1.5,8.35)  {$-\tfrac{1}{6}$};   
    \node at (1.3,7.6)  {$+\tfrac{2}{3}$};   
    \node at (1.3,6.3)  {$-\tfrac{2}{3}$};   
    \node at (1.3,5.35) {$-\tfrac{1}{6}$};   
    \node at (1.3,3.9) {$+\tfrac{1}{3}$};   
    \node at (1.3,3.1)    {$-\tfrac{1}{3}$};   
    \node at (1.3,1.6)  {$+1$};              

    \node[anchor=west] at (6.6,8) {$J_{r_1}=\tfrac{2}{3}$};
    \node[anchor=west] at (6.6,5) {$J_{r_3}=\tfrac{1}{6}$};
    \node[anchor=west] at (6.6,2) {$J_{r_5}=\tfrac{1}{3}$};

\end{tikzpicture}
\caption{The MASG flow $\f_{s,r} = \nu_{r,s} J_r(c_{\sigma,M})$ induced on the network of \fig{example-masg} by the perturbation $\sigma(E) = 1$, $M = \{A, C\}$ of \eq{example-perturbation}. Edge labels give the flow $\f_{s,r}$ along the directed edge $(s, r)$ (positive when flow runs from $s$ into reaction $r$, negative otherwise); the net reaction fluxes $J_{r_1}, J_{r_3}, J_{r_5}$ from \eq{example-fluxes} are annotated next to the reaction vertices. The MASG is $\sigma$--$M$ rigid, so these values are uniquely determined by Kirchhoff's law at the species vertices together with the stoichiometric ratios $\nu_{r,s}$ at the reaction vertices.}
\label{fig:example-flow}
\end{figure}

\textbf{Network parameters.} Substituting the fluxes into \eq{energy-masg}, the total
Gibbs free-energy consumption is
\begin{equation}
\Phi(c_{\sigma,M})
= \sum_{r \in \overrightarrow{{\cal R}}} \frac{J_r^2}{G_r}
= \tfrac{4}{9} + \tfrac{1}{36} + \tfrac{1}{9}
= \tfrac{7}{12} \approx 0.583,
\end{equation}
and the total network weight is
\begin{equation}
{\sf W} = \sum_{r \in \overrightarrow{{\cal R}}} \nu_r^2\, G_r
= 4 + 16 + 36 = 56.
\end{equation}
Since the perturbation $\sigma$--$M$ on this CRN is $\sigma$--$M$
rigid (the linear system at each reaction node has a unique solution
once stoichiometric ratios are imposed; see \sect{alt}), the
alternative effective resistance saturates the upper bound
${\sf R}_{\sigma,M}^{\sf alt} = {\sf E}(\theta) = \Phi(c_{\sigma,M})$, giving
\begin{equation}
\sqrt{{\sf R}_{\sigma,M}^{\sf alt}\,{\sf W}}
\;=\; \sqrt{\Phi \cdot {\sf W}}
\;=\; \sqrt{\tfrac{98}{3}} \;\approx\; 5.72.
\end{equation}

\textbf{Algorithmic cost.} Plugging into the cost expressions of \cor{detect-mas}, \cor{find-mas}, and \cor{approx-mas}, the four primitives evaluate
on this instance to the complexities summarised in \tabl{costs}.
\begin{table*}[t]
\centering
\caption{Quantum and classical costs for the main tasks.}
\label{tab:costs}
\begin{tabular}{lll}
\toprule
Task & Quantum cost & Classical baseline \\
\midrule
Reachability & $O({\sf S} + 5.72\,{\sf U}_\star)$
             & $\Omega(n^2)$ for $n = |\mathcal{S}| + |\overrightarrow{\mathcal{R}}| = 8$\\
Sampling     & $O({\sf S} + 5.72\,{\sf U}_\star \log^3|M|)$
             & --- \\
$\Phi$ estimation & $O\bigl(\tfrac{1}{\epsilon}({\sf S}
                  + \tfrac{1}{\epsilon}({\sf ET}^{\sf alt} + \log({\sf R}^{\sf alt}{\sf w}_s)){\sf U}_\star^{\sf alt})\bigr)$
                  & --- \\
Flow-state $\ket{\theta}$ & $O\bigl({\sf S} + \tfrac{1}{\epsilon^2}\bigl(\sqrt{{\sf ET}^{\sf alt}}
                  + \log({\sf R}^{\sf alt}{\sf w}_s)\bigr){\sf U}_\star^{\sf alt}\bigr)$ & --- \\
\bottomrule
\end{tabular}
\end{table*}

For this CRN the dissipation-aware bound $\sqrt{\Phi\,{\sf W}}\approx 5.72$ is
substantially smaller than the generic graph-size bound
$O(n^{3/2}) = O(8^{3/2}) \approx 22.6$ that the same algorithm would
incur if it ignored the chemistry. The gap is governed entirely by the
ratio $\Phi(c_{\sigma,M})/n$ — i.e.\ by how concentrated the
perturbation is — and grows as the CRN becomes larger and the
perturbation remains localised. This is the regime in which our
framework is most useful.

\section{Algorithms for the Gibbs free-energy consumption}\label{sec:gibbs}

At a first glance, one might hope to approximate the Gibbs free energy consumption $\Phi(c_{s,M})$ using \thm{effective}, since \eq{energy-masg} guarantees that for the MASG flow we have
$$ {\sf R}_{s,M} \leq {\sf E}(\f) = \Phi(c_{s,M}).$$
However, the flow $\f$ constructed in \eq{energy-masg} does not, in general, coincide with the true $\sigma$-$M$ electrical flow. As a result, we typically have ${\sf R}_{s,M} < \Phi(c_{s,M})$. Consequently, applying \thm{effective} yields only an approximation to a lower bound on $\Phi(c_{s,M})$, and the tightness of this bound is unclear. In this section, we show how we resolve this problem in certain cases.

\subsection{Multidimensional electrical networks}

Ref.~\cite{jeffery2023multidimensional} introduces the multidimensional quantum walk framework to generalise the standard quantum walk framework, later further formalised in~\cite{jeffery2025multidimensional}. A key innovation in this generalisation is the introduction of \textit{alternative neighbourhoods}.

\begin{definition}[Alternative Neighbourhoods]\label{def:alternative}
	For a network $G= (V,E,\w)$, as in \defin{network}, a set of \emph{alternative neighbourhoods} is a collection of states:
	$$\Psi_\star=\{\Psi_\star(u)\subset\mathrm{span}\{\ket{u,v}:v\in \N(u)\}: u\in V\}$$
	such that for all $u\in V$, $\ket{\psi_\star(u)} \in \Psi_\star(u)$ (see \eq{star-state}).
	We view the states of $\Psi_\star(u)$ as different possibilities for $\ket{\psi_\star(u)}$, only one of which is ``correct''.
\end{definition}

By adding additional alternative neighbourhoods to the set $\Psi_{\cal A}$ spanning ${\cal A}$ (see \eq{star-space}), we obtain $\Psi_{\Aalt}$. In the rest of this section, we assume that each $\Psi(u)$ is a normalised basis $\{\ket{\psi_{u,0}} = \ket{\psi_\star(u)},\dots,\ket{\psi_{u,a_u-1}}\}$, so
$$ \Psi_{\Aalt} = \{\ket{{\psi}_{u,i}}: u \in V \backslash \left(\{s\} \cup M\right), i \leq a_u-1\}.$$
Through the addition of these alternative neighbourhoods, we obtain the modified quantum walk operator
\begin{equation}\label{eq:walk-alt}
    U_{\Aalt\B}=(2\Pi_{\Aalt}-I)(2\Pi_{{\cal B}}-I),
\end{equation}
where $\Pi_{\Aalt}$ denotes the orthogonal projector onto the modified star space $\Aalt$.

Incorporating these additional neighbourhoods into $\Psi_\star(u)$ alters the quantum walk operator, and thus the complexity of its implementation. This was central to their original introduction in~\cite{jeffery2023multidimensional}, as these constructions are necessary in applications where it is computationally much more efficient to generate the set $\Psi_\star(u)$ rather than the individual star state $\ket{\psi_\star(u)}$. For example, this situation arises when $\ket{\psi_\star(u)}$ is known to be one of a small collection of easily preparable states within $\Psi_\star(u)$, but it is computationally infeasible to determine precisely which one. We denote the cost of implementing $U_{\Aalt\B}$ by ${\sf U}_{\star}^{\sf alt}$.

It was later shown in~\cite{li2025multidimensional} that the addition of alternative neighbourhoods structure in this way gives rise to the $\sigma$-$M$ alternative electrical flow associated with $\Psi_\star$. This flow satisfies the \emph{Alternative Kirchhoff's Law} and \emph{Alternative Ohm's Law}, which generalise \defin{kcl} and \defin{ohm}, respectively.

\begin{definition}[Kirchhoff's Alternative Law]\label{def:kcl-alt}
	For any \sm alternative flow $\falt$ with respect to a collection of alternative neighbourhoods $\Psi_{\star}$ on an electrical network $G = (V,E,\w)$, the corresponding flow state $\ket{\falt}$ is orthogonal to $\mathrm{span}(\Psi_{\star}(u))$ for every $u \in V \backslash \{\{s\}\cup M\}$, that is, $\langle {\psi}_{u,i} |\f \rangle =0$ for each $i\in \{0,1,\cdots,a_u-1\}$.
\end{definition}

\begin{definition}[Alternative Electrical Flow]\label{def:flow-alt}
	For a collection of alternative neighbourhoods $\Psi_{\star}$ on an electrical network $G = (V,E,\w)$, the \emph{\sm alternative electrical flow} is the unique alternative unit \sm flow $\falt$ with minimal energy ${\sf E}(\falt)$. We call this minimal energy the \emph{alternative effective resistance} ${\sf R}_{s,t}^{\sf alt}$.
\end{definition}

\begin{definition}[Alternative Ohm's Law]\label{def:ohm-alt}
	Let $\falt$ be the \sm alternative electrical flow with respect to a collection of alternative neighbourhoods $\Psi_{\star}$ on an electrical network $G = (V,E,\w)$. Then there exists an \textit{alternative potential vector} $\palt$ corresponding to $\falt$, which assigns a unique potential $\palt_{u,v}$ to each edge $(u,v) \in E$ such that $\palt_{s,v} = {\sf R}_{s,t}^{\sf alt}$ for every $v \in \N(s)$ and $\palt_{u,v} = 0$ for every $u \in M$ and $v \in \N(u)$. Additionally, the associated state $\ket{\palt}$ (see (16) in~\cite{li2025multidimensional}) satisfies $\Pi_{\Aalt}\ket{\palt} = \ket{\palt}$ and the potential difference between $(u,v)$ and $(v,u)$ is equal to the amount of electrical flow $\falt_{u,v}$ along $(u,v)$ multiplied with the resistance $1/\w_{u,v}$, that is, $\palt_{u,v} - \palt_{v,u} = \falt_{u,v}/\w_{u,v}$.
\end{definition}

The alternative potential $\palt$ from \defin{ohm-alt} gives rise to an alternative escape time, generalising \eq{escape}:
\begin{equation}\label{eq:escape-alt}
	\mathsf{ET}^{\sf alt}_{s,M} := \frac{1}{{\sf R}_{s,M}^{\sf alt}} \sum_{{u,v} \in E}\p_{u,v}^2\w_{u,v}.
\end{equation}
By using these ``alternative'' electrical-network definitions, Ref.~\cite{li2025multidimensional} shows that we can approximate the alternative electrical flow, generalising \thm{flow} to alternative neighbourhoods:

\begin{theorem}[\cite{li2025multidimensional}]\label{thm:flow-alt}
	Fix a network $G = (V,E,\w)$ as in \defin{network}, a collection of alternative neighbourhoods $\Psi_{\star}$, a non-empty marked set $M \subset V$ and an initial vertex $s \in V \setminus M$, with the promise that there is a path in $G$ connecting $s$ to $M$. Let $\falt$ be the \sm alternative electrical flow on $G$ with corresponding flow state $\ket{\falt}$ as defined in \eq{flowstate}. Then there exists a quantum walk algorithm that returns a state $\ket{\widetilde{\f}}$ satisfying
	$$ \frac{1}{2}\norm{\proj{\widetilde{\f}} - \proj{\falt}} \leq \epsilon$$
    with a constant probability of success in cost
    $$O\left({\sf S} + \frac{1}{\epsilon^2}\left(\sqrt{\mathsf{ET}_{s,M}^{\sf alt}} + \log({\sf R}_{s,M}^{\sf alt}\w_s) \right){\sf U}_{\star}^{{\sf alt}}\right).$$
\end{theorem}

\subsection{Using alternative neighbourhoods to restrict the flow}\label{sec:alt}

The technique of alternative neighbourhoods has previously been employed in~\cite{jeffery2023multidimensional, li2025multidimensional, jeffery2025multidimensional} to address cases where it is computationally easier to generate $\Psi_{\star}(u)$, which includes additional alternative neighbourhoods, than to construct the star state $\ket{\psi_{\star}(u)}$ directly. In this section, we exhibit a novel use of this technique: we show how alternative neighbourhoods can be leveraged to approximate the Gibbs free-energy consumption, and to prepare a quantum state that is a superposition over all species--reaction pairs $(r, s)$, weighted by their contribution to the total Gibbs free-energy consumption.

Our approach constructs a collection of alternative neighbourhoods $\Psi_{\star}$ such that the MASG flow defined in \eq{masgflow} is the unique \sm flow satisfying Alternative Kirchhoff's Law (see \defin{kcl-alt}). By construction (see \defin{flow-alt}), this guarantees that the MASG flow is also the alternative \sm\ electrical flow with respect to $\Psi_{\star}$. This reverses the usual causal relationship between alternative neighbourhoods and the corresponding electrical flow: rather than beginning with a prescribed set $\Psi_{\star}$ and checking which flows satisfy Alternative Kirchhoff's Law, we start from the desired flow and reverse-engineer the alternative neighbourhoods $\Psi_{\star}$ so that this flow becomes the alternative \sm electrical flow relative to $\Psi_{\star}$. A worked example illustrating this construction on a 3-species CRN, including the explicit alternative neighbourhood at a single reaction vertex and a comparison between the MASG flow and the (different) minimum-energy electrical flow, is provided in \app{alt-example}.

We now describe the general construction. We start with $r \in \R$, where the normalised projection of $\ket{\f}$ (where $\f$ is our MASG flow) onto $\textrm{span}\{\ket{r,s}: s \in {\cal S}\}$ is equal to
\begin{equation}\label{eq:flow-r}
    \ket{\f_r} \propto \sum_{s \in {\cal S}}\nu_{r,s}\ket{r,s}.
\end{equation}
As discussed in \sect{cost}, it is reasonable to assume that this state can be prepared efficiently in practical settings: the ratios of stoichiometric coefficients for the species involved in each reaction can be pre-computed classically and stored in QRAM.

In principle, performing the same preparation for every $s \in \mathcal{S}\setminus M$ would always reproduce the MASG flow. Unfortunately, this is not computationally viable in general. For $s \in \mathcal{S}\setminus M$, the normalised projection of $\ket{\f}$ onto $\operatorname{span}\{\ket{r,s}: r \in \mathcal{R}\}$ is
\begin{equation}\label{eq:flow-s}
    \ket{\f_s} \propto \sum_{r \in \R}\nu_{r,s}J_r(c_{\sigma,M})\ket{r,s}.
\end{equation}
Preparing this state requires knowledge of the relative fluxes $J_r(c_{\sigma,M})$ for all $r \in \mathcal{R}$. These fluxes depend on the steady-state concentrations $c_{\sigma,M}$, which can be obtained only by solving a global system of polynomial equations for the entire MAS, a task that is computationally infeasible. Consequently, we cannot generate the state in \eq{flow-s} and must restrict our alternative neighbourhood construction to vertices in $\mathcal{R}$.

The example in \app{alt-example} shows that this restriction \textit{can} be sufficient to recover the MASG flow uniquely, but it is not guaranteed for all networks. We therefore confine our attention to the following restricted class of networks:

\begin{definition}[$\sigma$-$M$ Rigid Network]\label{def:rigid}
Let $G = (V, E, \w)$ be a network as in \defin{network}, let $M \subset V$ be a non-empty marked set, and let $\sigma$ be a probability distribution supported on $V \setminus M$. We say that $G$ is \emph{\sm rigid} if the following conditions hold:
\begin{itemize}
    \item $G$ is bipartite with vertex partition $V = V_A \sqcup V_B$,
    \item $\sigma$ is supported entirely on $V_A$,
    \item For each $b \in V_B$, we fix a ratio vector $\rho_b : V_A \to \mathbb{R}_{\neq 0}$, supported on the neighbours of~$b$,
    \item Under these conditions, there exists a \emph{unique} \sm flow $\f$ such that for all $b \in V_B$ and all $a_1, a_2 \in V_A$ such that $(b, a_1), (b, a_2) \in E$ and $\rho_b(a_1), \rho_b(a_2)$,
    \[
    \frac{\f(b, a_1)}{\rho_b(a_1)} = \frac{\f(b, a_2)}{\rho_b(a_2)}.
    \]
\end{itemize}
\end{definition}

Provided the MASG is \sm rigid, we can apply our alternative neighbourhood technique to ensure that the MASG flow is the \emph{unique} \sm flow satisfying Alternative Kirchhoff's Law.

\subsection{Approximating and sampling from the Gibbs free-energy consumption}

Our quantum algorithm to estimate the Gibbs free-energy consumption is an extension of \thm{effective} that accommodates the use of alternative neighbourhoods:

\begin{theorem}[\cite{apers2022elfs}]\label{thm:effective-alt}
    Fix a network $G = (V,E,\w)$ as in \defin{network}, a collection of alternative neighbourhoods $\Psi_\star$, a non-empty marked set $M \subset V$ and an initial vertex $s \in V \setminus M$, with the promise that there is a path in $G$ connecting $s$ to $M$. Then there exists a quantum walk algorithm that $\epsilon$-multiplicatively estimates ${\sf R}_{s,M}^{\sf alt}$ with a constant probability of success in cost $$O\left(\frac{{\sf S}}{\epsilon} + \frac{1}{\epsilon^2}\left(\mathsf{ET}_{s,M}^{\sf alt} + \log({\sf R}_{s,M}^{\sf alt}\w_s) \right){\sf U}_{\star}^{\sf alt}\right).$$
\end{theorem}

The proof combines the strategy of \thm{effective} from~\cite{apers2022elfs} with Theorem~1.1 and Lemma~2.12 of~\cite{li2025multidimensional}, together with quantum amplitude estimation~\cite{brassard2002quantum, fukuzawa2023modified}; we provide it in \app{proof-effective-alt}.

Once the alternative \sm electrical flow is forced to coincide with the MASG flow, we can estimate its energy, namely the quantity $\Phi(c_{s,M})$ defined in \eq{energy-masg}. Note that the quantity $\w_s$ in this case is equal to $\sum_{r \in \R: s \in r}\nu_r \abs{\nu_{r,s}} G_r$, which can be precomputed with the thermodynamical information stored in the (Q)RAM. For ease of notation however, we shall keep it abbreviated as $\w_s$.

\begin{corollary}\label{cor:approx-mas}
    Fix an MAS \(({\cal S}, {\cal C}, {\cal R}, k)\) that admits a detailed balance, that is particle conserving and whose underlying CRN is reversible. Fix an initial species $s$ and a non-empty set of target species $M$ and assume that the resulting MASG is \sm rigid. Then there exists a quantum walk algorithm that $\epsilon$-multiplicatively estimates the instantaneous rate of Gibbs free-energy consumption $\Phi(c_{s,M})$ with a constant probability of success in cost $$O\left(\frac{1}{\epsilon}\left({\sf S} + \frac{1}{\epsilon}\left(\mathsf{ET}_{s,M}^{\sf alt} + \log\left(\Phi(c_{s,M})\w_s\right)\right){\sf U}_{\star}^{\sf alt}\right)\right).$$
\end{corollary}

Additionally, by employing \thm{flow-alt}, we can approximate, and subsequently sample from, the state
\[
\ket{\f} = \frac{1}{\sqrt{2\Phi(c_{s,M})}}\sum_{s \in {\cal S}}\sum_{r \in \R}\frac{\f_{s,r}}{\sqrt{\w_{s,r}}} \left(\ket{s,r} + \ket{r,s}\right).
\]

\begin{corollary}\label{cor:energy-mas}
	Fix an MAS \(({\cal S}, {\cal C}, {\cal R}, k)\) that admits a detailed balance, that is particle conserving and whose underlying CRN is reversible. Fix an initial species $s$ and a non-empty set of target species $M$ and assume that the resulting MASG is \sm rigid. Let $\f$ be the MASG flow (see \eq{masgflow}). Then there exists a quantum walk algorithm that returns a state $\ket{\widetilde{\f}}$ satisfying
	$$ \frac{1}{2}\norm{\proj{\widetilde{\f}} - \proj{\f}} \leq \epsilon$$
    with a constant probability of success in cost
    $$O\left({\sf S} + \frac{1}{\epsilon^2}\left(\sqrt{\mathsf{ET}_{s,M}^{\sf alt}} + \log(\Phi(c_{s,M})\w_s) \right){\sf U}_{\star}^{{\sf alt}}\right).$$
    Additionally, we can return an $\epsilon$-approximation of $J_r(c_{\sigma,M})^2/G_r$ for any $r \in \R$ with a constant probability of success in cost
    $$O\left(\frac{{\sf S}}{\epsilon} + \frac{1}{\epsilon^3}\left(\sqrt{\mathsf{ET}_{s,M}^{\sf alt}} + \log(\Phi(c_{s,M})\w_s) \right){\sf U}_{\star}^{{\sf alt}}\right).$$
\end{corollary}
\begin{proof}
    The first part of the corollary follows directly from applying \thm{flow-alt} with our modified quantum walk operator $U_{\Aalt\B}$, ensuring that the MASG flow (see \eq{masgflow}) matches the alternative electrical flow.

    For the second part, note that measuring the resulting state $\ket{\widetilde{\f}}$ yields an edge adjacent to reaction $r$ with a probability that is $\epsilon$-multiplicatively approximate to
    \[
    \frac{1}{\Phi(c_{s,M})}\sum_{s \in {\cal S}}\frac{\f_{s,r}^2}{\w_{s,r}} = \frac{J_r(c_{\sigma,M})^2}{G_r}.
    \]
    We can therefore approximate this quantity by applying quantum amplitude estimation.
\end{proof}

\subsection{Discussion and Outlook}

Many state-of-the-art algorithms for CRN analysis employ classical graph primitives such as Breadth-First Search, Depth-First Search, and random walks~\cite{Turtscher2022pathfinder, Ismail2022graphCRN_review, Margraf2019systematic} (see \app{crn-models} for a survey), but a formal computational cost analysis of these methods has been largely absent from the literature. This is primarily because a standard computational access model for CRNs has not been established. For the large-scale networks relevant in practical applications, however, it is reasonable to model access to the CRN's structure and its precomputed thermodynamic properties in a manner analogous to an adjacency matrix, with the network data stored in readable QRAM. In this access model, the species-reachability problem of \cor{detect-mas} reduces to \st\ connectivity, which classical algorithms require $\Omega(n^2)$ queries to solve, while quantum walks solve it in $O(n^{3/2})$ queries~\cite{durr2006quantum}, with $n$ the number of species. Our cost bound additionally incorporates the dissipation $\Phi(c_{\sigma,M})$ and the total MASG weight ${\sf W}$, so the speedup tightens whenever the perturbation is concentrated.

Our complexity statements rely on QRAM access to the precomputed thermodynamic data of the CRN, namely the stoichiometric coefficients $\nu_{r,s}$ and the Onsager coefficients $G_r$ at the equilibrium $c^*$. Whether large QRAMs are practical to fault-tolerantly implement remains an active question. Three features of the present setting partially mitigate this concern. First, the input is \emph{static}: once the equilibrium thermodynamics of a CRN has been computed, the same QRAM is reused across many perturbations $\eta$, amortising its construction cost. Second, the data is \emph{sparse}, as each reaction involves $O(1)$ species, so the MASG has $O(|{\cal R}|)$ edges rather than $O(|{\cal S}|^2)$, and the adjacency oracle reduces to a sparse-matrix oracle. Third, for the application classes most relevant to our framework, such as competitive enzyme inhibition, small metabolic motifs, and catalytic cycles, the precomputed dataset is small enough that the practical access cost is governed by classical I/O rather than by generic QRAM hardware.

Our analysis is predicated on three thermodynamic constraints: reversibility, the existence of a positive detailed-balance equilibrium $c^*$, and particle conservation. Many mechanistic models violate at least one of these: combustion schemes contain intrinsically irreversible steps, metabolic networks include committed reactions, and numerous biological or catalytic processes are deliberately driven far from equilibrium. A prominent class of CRNs that does satisfy all three is \emph{competitive binding}, in which an injected inhibitor species competes with the substrate for binding to an enzyme. This scenario fits naturally within the bipartite MASG framework: the additional species and binding/unbinding reactions preserve the network topology while modifying the equilibrium composition and the flux distribution~\cite{Pavlopoulos2018bipartite}. Kinetically, competitive inhibition increases the apparent Michaelis constant $K_M$ without affecting the maximal velocity $V_{\max}$ under the quasi-steady-state approximation~\cite{Cornish2013FundamentalsEnzyme, Segel1989quasiQSSA}, and the underlying expansion of the stoichiometric matrix is consistent with the deficiency-based framework~\cite{Feinberg2019foundations} and with mechanistic encodings widely used in systems biology~\cite{Klipp2016systems}. Although we do not explicitly model enzyme--inhibitor dynamics here, our framework applies directly to predicting how an inhibitor injection redistributes flux across a metabolic motif, which is a problem of practical interest in drug discovery.

Beyond the chemistry application, our use of alternative neighbourhoods is a methodological contribution to the multidimensional quantum-walk programme of \cite{jeffery2023multidimensional, jeffery2025multidimensional, li2025multidimensional}. In those works, alternative neighbourhoods were introduced as a tool to ease state preparation when the canonical star state $\ket{\psi_\star(u)}$ is hard to prepare directly; our use is structurally different, in that we exploit them to encode a local linear constraint on the flow at each reaction vertex, so that the only valid alternative flow is the chemically correct (non-minimum-energy) MASG flow. This suggests that constraint-encoding via local neighbourhoods is a general tool, applicable to any electrical-network problem in which a linear constraint at each vertex selects a unique flow from the affine family of \ssm flows. Natural directions for follow-up work include extending the framework to stochastic kinetics, time-dependent driving, and irreversible reaction networks, identifying CRN problems where structured data access can replace generic QRAM, and developing further applications of constraint-encoding alternative neighbourhoods to flow problems beyond chemistry. More broadly, our results indicate that quantum walks may provide a scalable approach for studying mechanistic aspects of biochemical regulation, drug action, and energy dissipation in large molecular networks.

\section{Acknowledgements}

This work was supported by the Dutch National Growth Fund (NGF), as part of the Quantum Delta NL programme.

\bibliographystyle{quantum}
\bibliography{main}

\clearpage
\onecolumn
\appendix

\section{A worked example: an s-t electrical network} \label{app:worked-electrical}

To make the definitions of \sect{graph-elec} concrete, consider the network $G = (V, E, \w)$ with vertex set $V = \{s, x, y, t\}$ and directed edge set $\E = \{(s, x), (x, y), (x, t), (y, t)\}$. The weight of each edge $(u, v) \in \E$ is $\w_{u,v} = \tfrac{1}{4}$, except for the edge $(s, x)$, which has weight $\w_{s,x} = 1$. This network is visualised in \fig{normal}, along with the \st\ electrical flow $\f$ on $G$ and the corresponding potential vector $\p$.

It is straightforward to verify that the flow $\f$ satisfies Kirchhoff's Law (see \defin{kcl}): the net flow entering any vertex other than the source $s$ and the sink $t$ is zero. The potential vector $\p$ satisfies Ohm's Law (see \defin{ohm}): for each edge $(u, v)$, the potential difference $\p_u - \p_v$ equals $\f_{u,v}/\w_{u,v}$. The effective resistance ${\sf R}_{s,t}$ can be determined in two equivalent ways: by computing the energy of the flow depicted in \fig{normal}, or by reading off the potential at $s$, $\p_s = \tfrac{11}{3}$. Either way, ${\sf R}_{s,t} = \tfrac{11}{3}$.

\begin{figure}[h]
	\centering
	\begin{tikzpicture}[scale=0.89, transform shape]
		\node at (0,0) {\begin{tikzpicture}
				\draw[->] (-2,0)--(-0.1,0);
				\draw[->] (0,0)--(1.9,0);
				\draw[->] (1.08,.92)--(1.92,.08);
				\draw[->] (0,0)--(.92,.92);

				\filldraw (-2,0) circle (.1);
				\filldraw (0,0) circle (.1);
				\filldraw (2,0) circle (.1);
				\filldraw (1,1) circle (.1);

				\node at (-2,-0.3) {$s$};
				\node at (0,-0.3) {$x$};
				\node at (1,1.25) {$y$};
				\node at (2,-0.3) {$t$};

				\node at (-1,0.25) {$1$};
				\node at (1,0.3) {$\frac{1}{4}$};
				\node at (0.3,0.7) {$\frac{1}{4}$};
				\node at (1.7,0.7) {$\frac{1}{4}$};

				\node at (0,-1) {$\w_{u,v}$ for each $(u,v) \in \E$};
		\end{tikzpicture}};

		\node at (6,0) {\begin{tikzpicture}
				\draw[->] (-2,0)--(-0.1,0);
				\draw[->] (0,0)--(1.9,0);
				\draw[->] (1.08,.92)--(1.92,.08);
				\draw[->] (0,0)--(.92,.92);

				\filldraw (-2,0) circle (.1);
				\filldraw (0,0) circle (.1);
				\filldraw (2,0) circle (.1);
				\filldraw (1,1) circle (.1);

				\node at (-2,-0.3) {$s$};
				\node at (0,-0.3) {$x$};
				\node at (1,1.25) {$y$};
				\node at (2,-0.3) {$t$};

				\node at (-1,0.25) {$1$};
				\node at (1,0.3) {$\frac{2}{3}$};
				\node at (0.3,0.7) {$\frac{1}{3}$};
				\node at (1.7,0.7) {$\frac{1}{3}$};
				\node at (0,-1) {$\f_{u,v}$ for each $(u,v) \in \E$};
		\end{tikzpicture}};

		\node at (12,0){\begin{tikzpicture}
				\draw[->] (-2,0)--(-0.1,0);
				\draw[->] (0,0)--(1.9,0);
				\draw[->] (1.08,.92)--(1.92,.08);
				\draw[->] (0,0)--(.92,.92);

				\filldraw (-2,0) circle (.1);
				\filldraw (0,0) circle (.1);
				\filldraw (2,0) circle (.1);
				\filldraw (1,1) circle (.1);

				\node at (-2,-0.3) {$s$};
				\node at (0,-0.3) {$x$};
				\node at (1,1.25) {$y$};
				\node at (2,-0.3) {$t$};

				\node at (-2,0.5) {$\frac{11}{3}$};
				\node at (-0,0.5) {$\frac{8}{3}$};
				\node at (1,0.5) {$\frac{4}{3}$};
				\node at (2,0.5) {$0$};
				\node at (0,-1) {$\p_u$ for each $u \in V$};
		\end{tikzpicture}};
	\end{tikzpicture}
	\caption{Graph $G$ with its \st\ electrical flow $\f$ and corresponding potential $\p$ at each vertex.}\label{fig:normal}
\end{figure}

\section{Worked computation of mass-action rate functions}\label{app:rate-example}

To illustrate \defin{kin}, we compute two of the rate functions for the example CRN of \eq{example-crn}. Recall the convention $0^0 = 1$ for any concentration $c_s = 0$.

For the reaction $A + C \to D$, the reactant complex is $y = A + C$, so $y_A = y_C = 1$ and $y_B = y_D = y_E = 0$. Hence
\[
{\cal K}_{A + C \to D}(c)
=
k_{A + C \to D} (c_A)^1 (c_B)^0 (c_C)^1 (c_D)^0 (c_E)^0
=
k_{A + C \to D}\, c_A c_C.
\]
For the reaction $2B \to A$, the reactant complex is $y = 2B$, so $y_B = 2$ and the remaining components vanish. Hence
\[
{\cal K}_{2B \to A}(c)
=
k_{2B \to A} (c_B)^2.
\]
In the abbreviated notation $\prod_{s \in {\cal S}} c_s^{y_s} = c^y$, both expressions reduce to ${\cal K}_{y \to y'}(c) = k_{y \to y'} c^y$.

\section{Linear-response derivation of the chemical affinity}\label{app:linear-response}

Here we provide the linear-response derivation underlying \eq{chempot-change}, \eq{netflux-change}, and \eq{netflux-change2}, following the conventions in~\cite{de2013non}. Recall the chemical potential
\(
\mu_s(c) = RT \ln(c_s) + \mu_s^0
\)
from \eq{chempot}. Since we are only interested in changes of $\mu_s(c)$ under a perturbation $\delta c_s := c_s - c^{*}_s$ around the equilibrium $c^*$, the reference $\mu_s^0$ drops out. Applying the first-order Taylor expansion of $\ln(c_s)$ around $c^*_s$ gives
\[
\delta\ln(c_s) = \ln(c^{*}_s+\delta c_s)-\ln(c^{*}_s)\approx\frac{\delta c_s}{c^{*}_s},
\]
so that
\begin{equation*}
  \delta \mu_s(c) =  RT\,\delta\ln(c_s) \approx  RT\,\frac{\delta c_s}{c^{*}_s}.
\end{equation*}
It is standard in linear non-equilibrium thermodynamics (see e.g.\ Chapter~X in~\cite{de2013non}) to assume this to be an equality, which is exact in the limit \(\delta c_s/c^{*}_s\to 0\). This recovers \eq{chempot-change}.

We can apply the same first-order expansion to $J_{r}(c)$, with $r = y \to y'$:
\begin{equation*}
\begin{split}
    \delta J_{r}(c) &= \sum_{s \in {\cal S}} \left. \frac{\partial J_{r}(c)}{\partial c_s}\right|_{c^{*}}\delta c_s \\
    &= \sum_{s \in {\cal S}} \left(\frac{y_s}{c^{*}_s}{\cal K}_{r}(c^*) - \frac{y'_s}{c^{*}_s}{\cal K}_{\overline{r}}(c^*)\right)\delta c_s \\
    &= -{\cal K}_{r}(c^*)\sum_{s \in {\cal S}} \frac{\nu_{r,s}}{c^{*}_s} \delta c_s,
\end{split}
\end{equation*}
where we used the detailed-balance condition ${\cal K}_{r}(c^*) = {\cal K}_{\overline{r}}(c^*)$ from \eq{db} in the last step. This recovers \eq{netflux-change}.

Combining the two expansions with the definitions of the affinity
\[
\Delta\mu_{r}(c) := -\sum_{s \in {\cal S}}\nu_{r,s}\, \delta\mu_s(c),
\]
and the Onsager coefficient
\[
G_{r} := \frac{{\cal K}_{r}(c^*)}{RT} > 0,
\]
yields the linear flux--affinity relation $J_r(c) = G_r \Delta\mu_r(c)$ stated in \eq{netflux-change2}. The instantaneous Gibbs free-energy consumption $\Phi(c) = \sum_{r \in \R} J_r(c)\Delta\mu_r(c)$ then reduces to $\Phi(c) = \sum_{r \in \R} J_r(c)^2/G_r$ near equilibrium, as in \eq{free-energy}.

\section{Quantum walks via phase estimation: framework details}\label{app:phase-estimation}

This appendix expands on \sect{qw} by giving the explicit phase-estimation construction underlying the four quantum walk theorems used in the main text. A quantum walk algorithm derived from the electrical-network framework can be viewed as a specialised phase-estimation algorithm~\cite{kitaev1996PhaseEst}, specified by the following parameters:
\begin{itemize}
    \item a finite-dimensional complex inner product space $H$,
    \item a unit vector $\ket{\psi_0}\in H$,
    \item two subspaces ${\cal A}, {\cal B}$ of $H$, such that $\ket{\psi_0}\in {\cal B}^{\perp}$.
\end{itemize}
These parameters define a quantum algorithm as follows. Let
\begin{equation}\label{eq:U-AB}
    U_{\cal{AB}}=(2\Pi_{\cal{A}}-I)(2\Pi_{\cal{B}}-I).
\end{equation}
We perform phase estimation of $U_{\cal{AB}}$ on initial state $\ket{\psi_0}$ to a certain precision, measure the phase register, and output $1$ if the measured phase is $0$, and $0$ otherwise. Here ``perform phase estimation'' means running a quantum subroutine that decides whether $\ket{\psi_0}$ has non-zero overlap with the $(+1)$-eigenspace of $U_{\cal{AB}}$, which, due to the structure of $U_{\cal AB}$, is precisely $({\cal A} \cap {\cal B}) \oplus ({\cal A}+{\cal B})^{\perp}$. Since by construction $\ket{\psi_0}$ is orthogonal to ${\cal B}$, the algorithm effectively decides whether $\ket{\psi_0} \in {\cal A} + {\cal B}$ or not.

The exact initialisation of the phase-estimation parameters depends on the application. In all our applications these parameters relate to objects from the classical random walk on the underlying graph and the electrical properties of the graph as an electrical network. For simplicity, we describe here the simplest initialisation, where $\textrm{supp}(\sigma) = \{s\}$ and the marked set $M$ is either empty or equal to $\{t\}$, and where the goal of the quantum walk is to decide which case is true (assuming $t$ is reachable from $s$).

The phase-estimation algorithm takes place on the edge space of $G$ from \eq{hilbert},
\[
{\cal H}=\mathrm{span}\{\ket{u,v}:\{u,v\} \in E(G)\},
\]
in which each edge $\{u,v\}$ ``appears twice'', as $\ket{u,v}$ and $\ket{v,u}$. The space ${\cal A}$ is built from the star states defined in \eq{star-state}:
\begin{equation}\label{eq:star-space}
    {\cal A} = \mathrm{span}\{\ket{\psi_\star(u)}: u \in V \setminus \{s,t\}\}.
\end{equation}
The space ${\cal B}$ is the antisymmetric subspace of ${\cal H}$,
\begin{equation}\label{eq:symm}
    {\cal B} = \mathrm{span}\{\ket{u,v} + \ket{v,u}:(u,v)\in E(G)\}.
\end{equation}
The initial state $\ket{\psi_0}$ is the star state of $s$ projected onto the symmetric subspace of ${\cal H}$ to ensure orthogonality with ${\cal B}$:
\begin{equation}\label{eq:state-init}
    \ket{\psi_0} = \sqrt{2} (I-\Pi_{\cal B})\ket{\psi_\star(s)}.
\end{equation}
Since by assumption $s$ and $t$ are connected, the \st\ electrical flow $\f$ exists, with corresponding flow state $\ket{\f}$ defined in \eq{flowstate} and living in ${\cal{B}}^\perp$. The key observation in~\cite{belovs2013quantum} is that, if $M = \{t\}$, the flow $\ket{\f}$ lies in $({\cal A} + {\cal B})^{\perp}$ and has non-zero overlap with $\ket{\psi_0}$, so $\ket{\psi_0} \notin {\cal A} + {\cal B}$ whenever $s$ and $t$ are connected. Conversely, if $M = \emptyset$, it is shown that $\ket{\psi_0} \in {\cal A} + {\cal B}$. The four quantum walk algorithms (\thm{detect}, \thm{find}, \thm{effective}, \thm{flow}) all build on this framework, modified to handle non-trivial initial distributions $\sigma$, multi-element marked sets $M$, and additional quantum subroutines beyond pure phase estimation; we refer the interested reader to the original sources for the full constructions.

\section{Worked example: alternative neighbourhoods on a 3-species CRN}\label{app:alt-example}

To build intuition for the construction described in \sect{alt}, we work through a small example. Consider the CRN
\begin{equation}\label{eq:crn-alt}
	\begin{split}
		A &\xrightleftharpoons[r_2]{r_1} B, \\
		A + B &\xrightleftharpoons[r_4]{r_3} 2C.
	\end{split}
\end{equation}
We let $\eta_{A} = -\eta_{C} = 1$ and show the resulting MASG and corresponding MASG flow (which is a unit flow from $A$ to $C$) in \fig{crn-alt}.

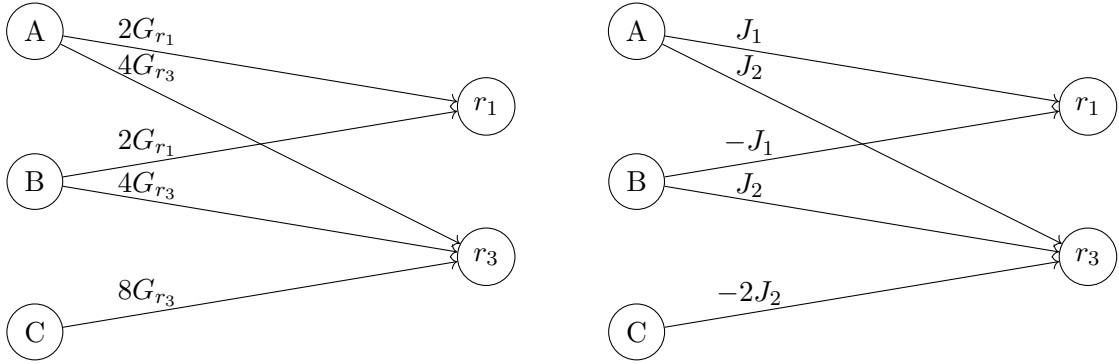
\begin{figure}
	\centering
	\begin{tikzpicture}
		\node at (0,0) {
			\begin{tikzpicture}
				\node[draw, circle] (R1) at (6,3) {$r_1$};
				\node[draw, circle] (R2) at (6,1) {$r_3$};

				\node[draw, circle] (A) at (0,4) {A};
				\node[draw, circle] (B) at (0,2) {B};
				\node[draw, circle] (C) at (0,0) {C};

				\draw[->] (A) to (R1);
				\draw[->] (B) to (R1);
				\draw[->] (A) to (R2);
				\draw[->] (B) to (R2);
				\draw[->] (C) to (R2);

				\node at (1.5,4) {$2G_{r_1}$};
				\node at (1.5,3.52) {$4G_{r_3}$};
				\node at (1.5,2.52) {$2G_{r_1}$};
				\node at (1.5,1.95) {$4G_{r_3}$};
				\node at (1.5,0.52) {$8G_{r_3}$};
			\end{tikzpicture}
		};
		\node at (8,0) {
			\begin{tikzpicture}
				\node[draw, circle] (R1) at (6,3) {$r_1$};
				\node[draw, circle] (R2) at (6,1) {$r_3$};

				\node[draw, circle] (A) at (0,4) {A};
				\node[draw, circle] (B) at (0,2) {B};
				\node[draw, circle] (C) at (0,0) {C};

				\draw[->] (A) to (R1);
				\draw[->] (B) to (R1);
				\draw[->] (A) to (R2);
				\draw[->] (B) to (R2);
				\draw[->] (C) to (R2);

				\node at (1.5,4) {$J_1$};
				\node at (1.5,3.52) {$J_2$};
				\node at (1.5,2.52) {$-J_1$};
				\node at (1.5,1.95) {$J_2$};
				\node at (1.5,0.52) {$-2J_2$};
			\end{tikzpicture}
		};
	\end{tikzpicture}
	\caption{The resulting MASG for the CRN in \eq{crn-alt}. The left image shows the weight assignments of each (directed) edge, the right image shows the corresponding flow values for the MASG flow, where for notational clarity we have omitted the argument $(c_{s,M})$ for each net flux $J_r$.}
	\label{fig:crn-alt}
\end{figure}

By \lem{masgflow}, we know that the MASG flow is a unit flow satisfying Kirchhoff's Law. In particular, this means that $J_1 = J_2 = \tfrac{1}{2}$. There is no guarantee that this flow matches the $A$-$C$ electrical flow on the MASG. In fact, for the simplified case where all $G_r = 1$, we see that the flow is not optimal, as shown in \fig{crn-opt}.

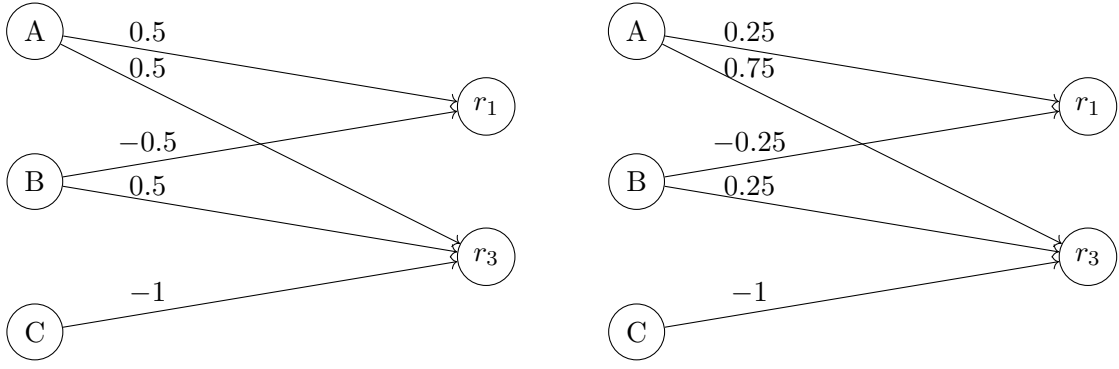
\begin{figure}
	\centering
	\begin{tikzpicture}
		\node at (0,0) {
			\begin{tikzpicture}
				\node[draw, circle] (R1) at (6,3) {$r_1$};
				\node[draw, circle] (R2) at (6,1) {$r_3$};

				\node[draw, circle] (A) at (0,4) {A};
				\node[draw, circle] (B) at (0,2) {B};
				\node[draw, circle] (C) at (0,0) {C};

				\draw[->] (A) to (R1);
				\draw[->] (B) to (R1);
				\draw[->] (A) to (R2);
				\draw[->] (B) to (R2);
				\draw[->] (C) to (R2);

				\node at (1.5,4) {$0.5$};
				\node at (1.5,3.52) {$0.5$};
				\node at (1.5,2.52) {$-0.5$};
				\node at (1.5,1.95) {$0.5$};
				\node at (1.5,0.52) {$-1$};
			\end{tikzpicture}
		};
		\node at (8,0) {
			\begin{tikzpicture}
				\node[draw, circle] (R1) at (6,3) {$r_1$};
				\node[draw, circle] (R2) at (6,1) {$r_3$};

				\node[draw, circle] (A) at (0,4) {A};
				\node[draw, circle] (B) at (0,2) {B};
				\node[draw, circle] (C) at (0,0) {C};

				\draw[->] (A) to (R1);
				\draw[->] (B) to (R1);
				\draw[->] (A) to (R2);
				\draw[->] (B) to (R2);
				\draw[->] (C) to (R2);

				\node at (1.5,4) {$0.25$};
				\node at (1.5,3.52) {$0.75$};
				\node at (1.5,2.52) {$-0.25$};
				\node at (1.5,1.95) {$0.25$};
				\node at (1.5,0.52) {$-1$};
			\end{tikzpicture}
		};
	\end{tikzpicture}
	\caption{The MASG flow for $J_1 = J_2 = \tfrac{1}{2}$ is shown on the left, and an $A$-$C$ flow with lower energy (assuming all $G_r = 1$) is shown on the right, disproving the fact that our MASG flow matches the $A$-$C$ electrical flow.}
	\label{fig:crn-opt}
\end{figure}

Recall from \eq{star-state} that the star state for the vertex $r_3$ is given by
\begin{align*}
	\ket{\psi_\star(r_3)} &= \frac{1}{\sqrt{\w_{r_3}}}\left(\sum_{v\in\N^+({r_3})}\sqrt{\w_{{r_3},v}}\ket{{r_3},v} - \sum_{v\in\N^-({r_3})}\sqrt{\w_{{r_3},v}}\ket{{r_3},v}\right) \\
	&= \frac{1}{2}\left(-\ket{r_3,A} - \ket{r_3,B} - \sqrt{2}\ket{r_3,C}\right).
\end{align*}
Now suppose that we add an alternative neighbourhood to $\Psi_\star(r_3)$:
\begin{equation}\label{eq:alt-example}
	\Psi_\star(r_3) = \left\{\frac{1}{2}\left(-\ket{r_3,A} - \ket{r_3,B} - \sqrt{2}\ket{r_3,C}\right), \frac{1}{\sqrt{2}}\left(\ket{r_3,A} - \ket{r_3,B}\right)\right\}.
\end{equation}
This means that $\Psi_\star(r_3)$ forms a basis for a 2-dimensional subspace of the 3-dimensional subspace spanned by the states $\{\ket{r_3,A},\ket{r_3,B},\ket{r_3,C}\}$. The flow state of any alternative $A$-$C$ flow satisfying Alternative Kirchhoff's Law must be orthogonal to this subspace, meaning in particular that
\begin{align*}
    \bra{\f}\frac{1}{2}\left(-\ket{r_3,A} - \ket{r_3,B} - \sqrt{2}\ket{r_3,C}\right) =
    -\frac{1}{2\sqrt{4G_{r_3}}}\left(\f_{r_3,A} + \f_{r_3,B} + \f_{r_3,C}\right) = 0,
\end{align*}
which is equivalent to the flow being conserved at $r_3$, and
\begin{align*}
    \bra{\f}\frac{1}{\sqrt{2}}\left(\ket{r_3,A} - \ket{r_3,B}\right) =
    \frac{1}{2\sqrt{4G_{r_3}}}\left(\f_{r_3,A} - \f_{r_3,B}\right) = 0,
\end{align*}
which forces $\f_{r_3,A} = \f_{r_3,B}$. Jointly, these constraints leave us only with a single option for the flow through $r_3$, since $\f$ must also be a unit flow from $A$ to $C$:
\[
2\f_{r_3,A} = 2\f_{r_3,B} = -\f_{r_3,C} = 1.
\]
Note that this immediately fixes the rest of the flow as well, since the outgoing flow of $A$ must be $1$, recovering the MASG flow.

\section{Proof of \thm{effective-alt} (alternative effective resistance)}\label{app:proof-effective-alt}

Here we provide the proof of \thm{effective-alt}, which we use in the main text to derive \cor{approx-mas}.

\begin{proof}[Proof of \thm{effective-alt}]
    The proof combines the strategy of the proof of \thm{effective} from~\cite{apers2022elfs} with Theorem~1.1 and Lemma~2.12 from~\cite{li2025multidimensional}. These theorems state that if we apply phase estimation with precision $O(\epsilon/({\sf ET}_{s,M}^{\sf alt}{\sf R}_{s,M}^{\sf alt}\w_s))$ to the modified quantum walk operator $U_{\Aalt\B}$ from \eq{walk-alt}, the probability of the phase estimation register yielding ``0'' is $p' = (1\pm \epsilon)/({\sf R}_{s,M}^{\sf alt}\w_s)$. Recall that $U_{\Aalt\B}$ is designed such that the MASG flow (see \eq{masgflow}) matches the alternative electrical flow, per \thm{flow-alt}.

    Hence by combining this with quantum amplitude estimation~\cite{brassard2002quantum, fukuzawa2023modified}, we obtain an $\epsilon$-multiplicative estimate of $p'$ at the cost of $O\left(\frac{1}{\epsilon}\log(\frac{1}{p'})\right)$ calls to the phase estimation routine.

    The proof is concluded by modifying the quantum walk with the standard quantum walk technique that ensures that ${\sf R}_{s,M}^{\sf alt}\w_s = O(1)$, incurring an extra $\log({\sf R}_{s,M}^{\sf alt}\w_s)$ (see~\cite{belovs2013quantum, apers2022elfs}).
\end{proof}

\section{Existing graph representations of CRNs}\label{app:crn-models}

Chemical reaction networks can be encoded as graphs in several inequivalent ways, each emphasising different structural and computational aspects. We briefly summarise three representative families of approaches that have been employed in the literature.

\subsection{Energy-weighted reaction graphs}

A first family~\cite{Haag2014interactive, Bergeler2015heuristics, Proppe2016uncertainty, Simm2017contextCRN, Simm2018error, Heuer2018integrated, Simm2018exploration, Proppe2018mechanism, Unsleber2020exploration, Baiardi2021expansive, Steiner2022autonomous, Unsleber2022chemoton, Turtscher2022pathfinder, Steiner2024human, Muller2024heron, Csizi2024nanoscale, Weymuth2024scine, Bensberg2024uncertainty} encodes a CRN as a weighted graph in which vertices represent stable molecular structures (typically minima on the relevant potential energy surface) and edges represent reactions, with edge weights derived from quantum chemical activation barriers or related reaction costs. Pathway analysis is then carried out by classical shortest-path algorithms such as Dijkstra's or Yen's algorithm; one representative implementation is the Pathfinder framework of~\cite{Turtscher2022pathfinder}. Methods in this family have evolved from interactive visualisation of reactivity, to heuristics-guided exploration, and more recently to error-controlled exploration using Gaussian-process regression and spline-based interpolation of reaction paths. The accuracy of the resulting analysis is bounded by the accuracy of the barrier data, and the cost of obtaining these barriers at quantum chemical accuracy currently limits scaling to large networks.

\subsection{Connectivity-matrix representations}

A second family~\cite{Habershon2015samplingRW, Habershon2016graphrxnpath, Ismail2019auto_multi_rxn_graph, Ismail2022graphCRN_review, Ismail2022success_challenge_ML, Robertson2019fast_screen_graph, Robertson2020dense_network_rxn, Fakhoury2023contactmap_directedwalks_protein_1, Fakhoury2024contactmap_directedwalks_protein_2} encodes each chemical species as a molecular graph, typically via a connectivity matrix. Reactions are then graph transformations corresponding to bond-making and bond-breaking events, governed by a predefined set of rules or templates. The full CRN becomes a graph whose vertices are connectivity patterns and whose edges are these transformations, and random walks on this graph provide a stochastic but scalable strategy for pathway sampling. Limitations of this approach include the neglect of conformational diversity, the difficulty of representing proton-coupled processes, and the exponential growth of the connectivity state space with molecular size.

\subsection{Bipartite molecule--reaction graphs}

A third family represents a CRN as a directed bipartite graph~\cite{Margraf2019systematic} in which one vertex partition corresponds to molecules and the other to reactions, with edges encoding participation. Directionality distinguishes reactants from products, and reversing the orientation of all edges incident to a reaction vertex leaves the representation unchanged, so reversibility is captured naturally. The construction generalises to arbitrary stoichiometry without modification and emphasises the indirect coupling between molecules induced by reactions, in the spirit of the ``virtual flask'' framework of~\cite{Simm2017contextCRN}. Bipartite graphs are also well suited to network-theoretic analysis, constraint-based modelling, and quantum-walk formulations.

In this work we adopt the bipartite representation as our baseline, since it cleanly separates species from reactions, is invariant under the choice of reaction orientation, and admits the electrical-network analogy on which our quantum walk algorithms rely (see \sect{masg}).

\end{document}